\documentclass[a4paper]{article}
\usepackage[11pt]{extsizes}
\usepackage[utf8x]{inputenc}
\usepackage[english]{babel}
\usepackage{amssymb, amsmath, amsthm, verbatim, cite, lscape, longtable, graphicx, wrapfig, tikz,url}

\usepackage{amsfonts,amsmath,mathtools,hyperref,cite}

\newtheorem{formula}{}[section]
\newtheorem{definition}[formula]{Definition}
\newtheorem{corollary}[formula]{Corollary}
\newtheorem{remark}[formula]{Remark}
\newtheorem{lemma}[formula]{Lemma}
\newtheorem{proposition}[formula]{Proposition}
\newtheorem{theorem}[formula]{Theorem}

\newcommand{\R}{\mathbb{R}}

\newcommand{\F}{\mathbb{F}}
\newcommand{\Fq}{\F_q}

\newcommand{\C}{\mathcal{C}}

\newcommand{\Bil}{\operatorname{Bil}}

\newcommand{\MatnmFq}{\spac}

\newcommand{\bfn}{\mathbf {n}}
\newcommand{\bfm}{\mathbf {m}}
\newcommand{\spac}{\F^{\bfn\times\bfm}_q}
\newcommand{\srk}{\mathrm{srk}}
\newcommand{\rk}{\mathrm{rk}}

\newcommand{\Del}{D(q,\bfn,\bfm,d)}

\DeclareMathOperator{\tr}{tr}

\newcommand{\mC}{\mathcal{C}}

\newcommand{\floor}[1]{{\left\lfloor{#1}\right\rfloor}}

\newlength{\mynodespace}
\usepackage{stmaryrd}
\newcommand{\db}[1]{\llbracket #1 \rrbracket}

\title{Semidefinite and linear programming bounds for sum-rank-metric codes and non-existence results} 
\author{
Aida Abiad
\thanks{\texttt{a.abiad.monge@tue.nl}, Department of Mathematics and Computer Science, Eindhoven University of Technology, The Netherlands}
\thanks{Department of Mathematics and Data Science, Vrije Universiteit Brussel, Belgium} 
\and
Antonina P. Khramova
\thanks{\texttt{a.khramova@tue.nl}, Department of Mathematics and Computer Science, Eindhoven University of Technology, The Netherlands} 
\and
Sven C. Polak
\thanks{\texttt{s.c.polak@tilburguniversity.edu},  Department of Econometrics and Operations Research, Tilburg University, The Netherlands} 
\and
Ferdinando Zullo
\thanks{\texttt{ferdinando.zullo@unicampania.it}, Dipartimento di Matematica e Fisica, Universit\`a degli Studi della Campania ``Luigi Vanvitelli'', Italy} 
}
\date{}
\begin{document}

\maketitle

\begin{abstract}
The sum-rank metric provides a unifying framework that generalizes both the celebrated Hamming and rank metrics, and has found applications in areas such as network coding, distributed storage, and space-time coding. A central problem is to determine the maximum size of a code with prescribed minimum distance.

In this paper, we derive new sharp upper bounds on the size of a sum-rank-metric code using spectral and optimization techniques, including a semidefinite programming (SDP) bound that can outperform the best existing bounds based on computational experiments. Furthermore, we compare the Delsarte linear programming (LP) bound and a recent eigenvalue LP bound, and show equivalences between them, with particular emphasis on extremal regimes of the sum-rank metric. Finally, we show how to use the several SDP, LP and eigenvalue bounds to prove non-existence results for certain optimal and perfect sum-rank metric codes. Our results suggest that the combination of spectral and optimization methods effectively captures the hybrid nature of the sum-rank metric, providing new techniques that overcome the limitations of classical coding-theoretic approaches.
\end{abstract}

\section{Introduction}

The sum-rank metric has recently emerged as a unifying framework in coding theory, bridging two of its most prominent settings: the Hamming metric and the rank metric. Originally introduced in the context of multishot network coding and coding for matrix channels, the sum-rank metric naturally models hybrid error patterns where both symbol-wise and rank-type errors occur simultaneously. This makes it particularly well-suited for modern applications such as network coding~\cite{silva2008rank,koetter2008coding} and space-time coding~\cite{yan2009error}; see \cite{martinez2022codes} for more on the applications. 

From a theoretical perspective, the study of sum-rank-metric codes generalizes classical problems in coding theory. Given a finite field $\mathbb{F}_q$ and parameters describing the block structure of the metric, a fundamental question is to determine the maximum cardinality of a code with prescribed minimum distance. While this problem has been extensively studied in the Hamming metric~\cite{macwilliams1977theory} and in the rank metric~\cite{D1978,gabidulin1985theory}, the sum-rank setting presents additional challenges due to its hybrid structure. In particular, the interaction between different blocks of the metric leads to new combinatorial and algebraic phenomena that are not present in the classical cases.

A powerful recent approach to tackling such coding problems is through connections with graph theory. To each metric space one can associate a graph whose vertices correspond to the ambient space and whose edges reflect adjacency with respect to the metric. In this framework, the problem of determining the largest code with given minimum distance is equivalent to bounding the independence number of the corresponding graph \cite{APR2025}. This viewpoint enables the use of spectral methods and optimization techniques, such as eigenvalue bounds and linear programming methods originally introduced by Delsarte~\cite{DelsarteLP}.

In the sum-rank setting, several bounds have been developed by adapting classical coding techniques; these include bounds induced from the Hamming metric, such as the Singleton, Hamming, Plotkin, and Elias bounds~\cite{byrne2021fundamental}, as well as bounds that exploit the specific structure of the sum-rank metric. Furthermore, spectral graph theory based approaches have led to the development of the Ratio-type bound and its linear programming implementation~\cite{AKR2024}. Recently, Delsarte’s linear programming method has recently been extended to the socalled sum-rank-metric association scheme, providing another powerful tool for obtaining the best known upper bounds on sum-rank-metric code sizes~\cite{AGKP2024}.

Although semidefinite programming (SDP) methods have been applied successfully to problems in coding theory—often producing some of the strongest known bounds in the Hamming and related metrics (see, e.g., \cite{S2005,bachoc2010applications})—their use in the context of sum-rank metric codes remains unexplored. Motivated by this and inspired by Schrijver's machinery \cite{S2005}, we  derive a new SDP upper bound for the size of sum-rank metric codes. We also show several equivalences between the existing LP bounds for sum-rank-metric codes, with a focus on their relative strength and structural implications. Our contributions can be summarized as follows:
\begin{itemize}
    \item We introduce a semidefinite programming (SDP) bound for the sum-rank metric, extending techniques previously applied to Hamming and related metrics. Using computational methods, we show that this bound improves upon existing bounds in several parameter regimes.
    \item We provide the first application of eigenvalue-based bounds to derive non-existence results for perfect sum-rank-metric codes. This demonstrates the effectiveness of spectral techniques in addressing structural questions in the sum-rank setting.
    \item We perform a detailed comparison between the Delsarte linear programming bound and the Ratio-type bound, with particular emphasis on extremal regimes of the sum-rank metric. In the rank-metric case, we establish the equivalence of these bounds, while in more general settings we highlight their differences.
\end{itemize}

As a consequence of the above results, we also obtain multiple non-existence results for both maximum sum-rank distance (MSRD) codes—i.e., codes attaining the Singleton bound—and perfect codes, i.e., codes attaining the Sphere-Packing bound with equality. These results show that, in many cases, the known upper bounds are strictly smaller than the values predicted by optimal constructions, thereby ruling out their existence, extending several known results on sum-rank codes existence \cite{byrne2021fundamental,AKR2024,PRZ2025}.

\paragraph{Structure of the paper.}
The paper is organized as follows. In Section \ref{sec:preliminaries}, we introduce the necessary preliminaries on sum-rank-metric codes, including notation, known bounds, and the combinatorial structure of spheres and balls. We also recall the connection between sum-rank-metric spaces and graph theory, and review eigenvalue-based bounds such as the Ratio-type bound and other optimization bounds such as Delsarte LP bound. In Section \ref{sec:equivalencesLPbounds}, we compare and show several new equivalences between the existing linear programming bounds for sum-rank codes in the extremal regimes (rank-metric and Hamming metric). Section \ref{sec:SchrijverSDPbound} is devoted to the introduction of a semidefinite programming (SDP) bound for the sum-rank metric, together with computational results demonstrating its improvements over all existing bounds. In Section \ref{sec:nonexistence}, we apply the eigenvalue Ratio-type bound and the new SDP bound to study the existence of extremal codes, deriving multiple new non-existence results for MSRD and perfect sum-rank-metric codes. Finally, Section \ref{sec:concludingremarks} contains some concluding remarks and future directions.

\section{Preliminaries and existing bounds}\label{sec:preliminaries}

Let us introduce the notation and basic concepts used throughout the paper. We begin by recalling the definition of sum-rank-metric spaces and codes, together with the main parameters of interest. We then provide a concise overview of the bounds currently known for sum-rank-metric codes, which will serve as a benchmark for the results developed in later sections. Finally, we review the combinatorial structure of spheres and balls in the sum-rank metric and recall the corresponding Sphere-Packing bound, which plays a central role in our analysis.

\subsection{Sum-rank-metric codes}

For a positive integer $t$, $[t]$ is used to denote the set $\{1,2,\dots,t\}$, while $\db{t}$~denotes the set $\{0\}\cup[t]$. Let ${\bf n}=(n_1,\ldots,n_t)$, ${\bf m}=(m_1,\ldots,m_t)$ be tuples of positive integers 
with $m_1 \geq m_2 \geq \cdots \geq m_t$, 
and $m_i\geq n_i$ for all $i\in [t]$.

For a prime power $q$ and 
positive integers $m\geq n$, 
let $\mathbb{F}_q^{n\times m}$
denote the vector space of all $n\times m$ matrices over the finite field $\mathbb{F}_q$. 
Denote by $\rk(M)$ the rank of a matrix 
$M\in \mathbb{F}_q^{n\times m}$.

 \begin{definition}\label{def:srk}
The \emph{\textbf{sum-rank-metric space}} is 
the $\mathbb{F}_q$-linear vector space 
$\mathbb{F}_q^{{\bf n}\times {\bf m}}$ defined as follows:
\[
\mathbb{F}_q^{{\bf n}\times {\bf m}}:=
\mathbb{F}_q^{n_1\times m_1}\times \cdots \times
\mathbb{F}_q^{n_t\times m_t},
\]
where $\times$ stands for the direct product of vector spaces.
The \emph{\textbf{sum-rank}} of an element $X=(X_1,\ldots, X_t)\in \mathbb{F}_q^{{\bf n}\times {\bf m}}$ 
is $\srk(X) = \sum_{i=1}^t \rk(X_i)$.
The \emph{\textbf{sum-rank distance (metric)}} between 
$X,Y \in \mathbb{F}_q^{{\bf n}\times {\bf m}}$ is $\srk(X - Y)$.
In case $t=1$ we sometimes refer to the sum-rank distance as simply the \emph{\textbf{rank distance}}, and the corresponding metric space as the \emph{\textbf{rank-metric space}}.
\end{definition}

It is easy to see that the sum-rank distance is indeed a distance on $\mathbb{F}_q^{{\bf n}\times {\bf m}}$.

\begin{definition}
A \emph{\textbf{sum-rank-metric code}} is a non-empty subset $C \subseteq \mathbb{F}_q^{{\bf n}\times {\bf m}}$. 
The \emph{\textbf{minimum sum-rank distance}} of a code $C$ with $|C| \geq 2$ is defined 
by $$\srk(C) := \min\left\{\srk(X-Y)\colon X, Y \in C, \, X \neq Y\right\}.$$

If $|C|\leq 1$, we define $\srk(C)=\infty$.
\end{definition}

Determining the maximum size of a code with prescribed minimum distance is a classical problem in coding theory, extensively studied in the Hamming metric. 
In this paper, we consider the analogous problem for the sum-rank metric. Specifically, we study bounds on
\[
A(q,\bfn,\bfm,d):=\max\{|C|\colon C\subseteq \MatnmFq,\; \srk(C)\ge d\}.
\]

\subsection{Overview of known coding bounds}\label{sec:code-bounds}

In~\cite{byrne2021fundamental}, the authors
observe that if $\mC$ is a sum-rank-metric code in $\spac$ of size at least $2$ with $\srk(\mC) \ge d$,
then the size of $\mC$ is upper bounded by the size of the largest Hamming-metric code over field $\F_{q^m}$ of length $N$ and minimum distance at least $d$, where
$m=\max\{m_1,\dots,m_t\}$.
This observation implies the following bounds.

\begin{theorem}[see~\text{\cite[Theorem III.1]{byrne2021fundamental}}]\label{thm:induced} Let $m=\max\{m_1,\dots,m_t\}$ and let $\mC\subseteq\spac$ be a sum-rank-metric code with $|\mC|\geq 2$ and $\srk(\mC) \ge d$. The following hold.
\begin{center}
    \begin{tabular}{ll}
     \textbf{Induced Singleton bound:} & $|\mC|\leq q^{m(N-d+1)}$, \\[0.5cm]
     \textbf{Induced Hamming bound:} & $|\mC|\leq \floor{\frac{q^{mN}}{\sum_{s=0}^{\floor{(d-1)/2}}\binom{N}{s}(q^m-1)^s}}$, \\[0.5cm]
     \textbf{Induced Plotkin bound:} & $|\mC|\leq \floor{\frac{q^{m}d}{q^md-(q^m-1)N}}$ if $d>(q^m-1)N/q^m$, \\[0.5cm]
     \textbf{Induced Elias bound:} & $|\mC|\leq \floor{\frac{Nd(q^m-1)}{q^mw^2-2Nw(q^m-1)+(q^m-1)Nd}\cdot\frac{q^{mN}}{V_w(\F^N_{q^m})}}$. \\[0.5cm]
\end{tabular}
\end{center}
In the Induced Elias bound, $w$ is any integer between $0$ and $N(q^m-1)/q^m$ such that the denominator is positive, and $V_w(\F_{q^m}^N)=\sum_{i=0}^w {N\choose i}(q^m-1)^i$, i.e. the cardinality of any ball of radius $w$ in $\F_{q^m}^N$ with respect to the Hamming distance.
\end{theorem}

The previous bounds can be derived by approximating the sum-rank metric with the Hamming metric. By exploiting the structure of the sum-rank metric more carefully, one obtains the following bounds, which are typically tighter.

\begin{theorem}[see \text{\cite[Theorems III.2 and III.8]{byrne2021fundamental}}]\label{thm:non-induced-singleton-td} Let $\mC\subseteq\spac$ be a code with $|\mC|\geq2$ and $\srk(\mC)\ge d$. Let $j$ and $\delta$ be the unique integers satisfying $d-1=\sum_{i=1}^{j-1} n_i + \delta$ and $0\leq\delta\leq n_j-1$.  Finally, let $Q=\sum_{i=1}^t q^{-m_i}$.
The following hold.
{\small{
\begin{center}
    \begin{tabular}{ll}
     \textbf{Singleton bound:} & $|\mC|\leq q^{\sum_{i=j}^t m_in_i-m_j\delta}$, \\[0.5cm]
     \textbf{Total Distance bound:} & $|\mC|\leq\frac{d-N+t}{d-N+Q}$ if $d>N-Q$. \\[0.5cm]
\end{tabular}
\end{center}
}}

\end{theorem}

The Singleton bound in Theorem~\ref{thm:non-induced-singleton-td} is the analogue of the classical Singleton bound for the Hamming metric. Codes attaining the Singleton bound in the Hamming metric are known as maximum distance separable (MDS) codes. Similarly, codes attaining the Singleton bound in the sum-rank metric are called \textbf{maximum sum-rank distance} (MSRD) codes.


\subsection{Spheres and balls in the sum-rank metric and the Sphere-Packing bound}\label{sec:spheresandperfectcodes}
In order to derive packing-type bounds for sum-rank-metric codes, we study the combinatorial structure of spheres and balls in the ambient space $\spac$. As in the classical Hamming and rank-metric settings, the size of these objects plays a fundamental role in bounding the maximum cardinality of a code with prescribed minimum distance. 

We begin by introducing spheres and balls in the sum-rank metric and recalling known expressions for their volumes. Later we recall the Sphere-Packing bounds for codes in $\spac$. 

\begin{definition}
Let $l$ be a non-negative integer and $X\in\spac$. The \textbf{sphere} of radius $l$ centered at $X$ in the sum-rank metric is
\[
S(X,l)=\{Y\in\spac \mid \srk(X-Y)=l\},
\]
the \textbf{ball} of radius $l$ centered at $X$ in the sum-rank metric is
\[
B(X,r)=\{Y\in\spac \mid \srk(X-Y)\le r\}.
\]
\end{definition}

Clearly,
\begin{equation}\label{c2}
B(X,r)=\bigcup_{l=0}^r S(X,l).
\end{equation}

Since the sum-rank metric is translation invariant, the size of spheres and balls does not depend on the center. Hence, we define the \textbf{volume} of a ball of radius $r$ as
\[
\mathbf{V}_r(\spac):=|B(0,r)|.
\]
Similarly, for $l\ge0$ we define
\[
S_l=\{Y\in\spac \mid \srk(Y)=l\}, \qquad 
\mathbf{V}(S_l):=|S_l|.
\]
From \eqref{c2} it follows that
\begin{equation}\label{c3}
\mathbf{V}_r(\spac)=\sum_{l=0}^{r}\mathbf{V}(S_l).
\end{equation}

The number of elements in $\spac$ having a fixed sum-rank weight was computed in \cite[Section III]{byrne2021fundamental} (see also \cite[Section 4]{sauerbier2025bounds} and \cite{PRZ2025}), yielding the following result.

\begin{proposition}\label{c6}
Let $l,r\ge0$. Then
\[
\mathbf{V}(S_l)=
\sum_{\substack{(k_1,\ldots,k_t)\in\mathbb{N}_0^t \\ k_1+\cdots+k_t=l}}
\prod_{i=1}^{t}
{n_i \choose k_i}_q
\prod_{j=0}^{k_i-1}(q^{m_i}-q^j),
\]
and
\[
\mathbf{V}_r(\spac)=
\sum_{l=0}^{r}
\sum_{\substack{(k_1,\ldots,k_t)\in\mathbb{N}_0^t \\ k_1+\cdots+k_t=l}}
\prod_{i=1}^{t}
{n_i \choose k_i}_q
\prod_{j=0}^{k_i-1}(q^{m_i}-q^j),
\]
\end{proposition}

where for integers $0 \leq j \leq n$,
\[
{n \choose j}_q := \prod_{i=0}^{j-1} \frac{q^{n-i} - 1}{q^{j-i} - 1},
\]
and is called \textbf{Gaussian binomial coefficient} (or $q$-binomial coefficient).
This quantity counts the number of $j$-dimensional subspaces of $\mathbb{F}_q^n$.

The above expression is not very convenient for computations, especially when studying covering properties; see also \cite{puchinger2022generic} for a dynamic programming approach to evaluate it.

We explicitly treat the special case in which all but one of the entries of $\mathbf{m}$ and $\mathbf{n}$ are equal to one.

\begin{proposition}\label{prop:boundvolume}
    In the sum-rank-metric space $\spac$ with $\bfm=(m,1,\dots,1)$, $\bfn=(n,1,\dots,1)$, where $m\geq n\geq1$, while $q\geq 2$ is a prime power, we have that
    \[ V_r(\Fq^{\bfn\times\bfm})= \sum_{i=0}^r\sum_{j=0}^i {n \choose j}_q (q^m-1)\cdots(q^m-q^{j-1}) {t-1\choose i-j}(q-1)^{i-j}\]
    \[\leq (r+1) {t-1 \choose \lfloor (t-1)/2\rfloor} (q-1)^r q^{(m+n+1)r-r^2+1}. \]
    In particular, $V_r \equiv \sum_{i=0}^r {t-2 \choose i}\pmod{q}$.
\end{proposition}

\begin{proof}
    We start by computing the number of elements in $\Fq^{\bfn\times\bfm}$ that have sum-rank $i$, for some $i\leq r$.
    An element $(X,x_1,\ldots,x_{t-1})\in \Fq^{\bfn\times\bfm}$ has sum-rank $i$ if and only if the Hamming weight of $(x_1,\ldots,x_{t-1})$ is $i-\mathrm{rk}(X)$ and so the number of elements in $\Fq^{\bfn\times\bfm}$ with sum-rank weight $i$ is 
    \[\sum_{j=0}^i {n \choose j}_q (q^m-1)\cdots(q^m-q^{i-1}) {t-1\choose i-j}(q-1)^{i-j}, \]
    where the first part counts the number of matrices in $\Fq^{m\times n}$ of rank $j$ and the second part counts the number of elements in $\F_q^{t-1}$ with Hamming weight $i-j$.
    From this we derive the assertion.
    Using \cite[Proposition 1]{loidreau2006properties}, we know that
    \[ {n \choose j}_q (q^m-1)\cdots(q^m-q^{i-1}) \leq q^{(m+n+1)j-j^2}, \]
    and so
    \begin{align*}
        V_r(\Fq^{\bfn\times\bfm})&\leq \sum_{i=0}^r\sum_{j=0}^i q^{(m+n+1)j-j^2} {t-1\choose i-j}(q-1)^{i-j} \\   
    & \leq \sum_{i=0}^r q^{(m+n+1)i-i^2} \sum_{j=0}^i {t-1\choose i-j}(q-1)^{i-j}\\
    &\leq \sum_{i=0}^r q^{(m+n+1)i-i^2} (i+1) {t-1 \choose i} (q-1)^i\\
   & \leq (r+1) {t-1 \choose \lfloor (t-1)/2\rfloor} (q-1)^r q^{(m+n+1)r-r^2+1}. 
    \end{align*}
    For the last part, observe that ${n \choose j}_q \equiv 1 \pmod{q}$ and so
    \[ V_r \equiv \sum_{i=0}^r \sum_{j=0}^i (-1)^j {t-1 \choose i-j} \equiv \sum_{i=0}^r {t-2 \choose i}\pmod{q}. \]
\end{proof}

We point out the following, which will be used later.

\begin{corollary}\label{cor:V1}
     In the sum-rank-metric space $\spac$ with $\bfm=(m,1,\dots,1)$, $\bfn=(n,1,\dots,1)$, where $m\geq n\geq1$, while $q\geq 2$ is a prime power, we have that
    \[ V_1(\Fq^{\bfn\times\bfm})= \frac{q^n-1}{q-1} (q^m-1)+ (t-1)(q-1)+1.\]
\end{corollary}

The classical approach used in the Hamming metric for proving the Sphere-Packing bound also works in the sum-rank metric, leading to the following result.

\begin{theorem}[see \text{\cite[Theorems III.6 and III.7]{byrne2021fundamental}}]\label{thm:non-induced} Let $\mC\subseteq\spac$ be a code with $|\mC|\geq2$ and $\srk(\mC)\ge d$. Let $\ell\leq t-1$ and $\delta'\leq n_{\ell+1}-1$ be the unique positive integers such that $\smash{d-3=\sum_{j=1}^\ell n_j+\delta'}$.
Define $\smash{\bfn'=(n_{\ell+1}-\delta',n_{\ell+2},\dots,n_t)}$ and $\bfm'=(m_{\ell+1},m_{\ell+2},\dots,m_t)$. 
The following hold.
{\small{
\begin{center}
    \begin{tabular}{ll}
     \textbf{Sphere-Packing bound:} & $|\mC|\leq\floor{\frac{|\spac|}{V_r(\spac)}}$, where $r=\floor{(d-1)/2}$, \\[0.5cm]
     \textbf{Projective Sphere-Packing bound:} & $|\mC|\leq\floor{\frac{|\Fq^{\bfn'\times\bfm'}|}{V_1(\Fq^{\bfn'\times\bfm'})}}$ if $3\leq d\leq N$.
\end{tabular}
\end{center}
}}

\end{theorem}

As classically done in the Hamming metric, we give the definition of perfect code.

\begin{definition}
    Let $\mC\subseteq\spac$ be a code with $|\mC|\geq2$ and $\srk(\mC)\ge d$. If its parameters reach the Sphere-Packing bound with equality, then we say that $\mC$ is a \textbf{perfect code}.
\end{definition}

\subsection{The sum-rank-metric graph and overview of known graph theoretical bounds}\label{sec:graph}

In previous section, an overview of graph theoretical bounds on the maximal cardinality of a sum-rank-metric code was provided. In this section, we explore the sum-rank metric from a graph-theoretical point of view, and discuss other known bounds that arise from inspecting the eigenvalues of an adjacency matrix of a respective graph.

\begin{definition}\cite{byrne2022anticodes} For a sum-rank-metric space $\spac$, the \textbf{sum-rank-metric graph} $\Gamma(\spac)$ is a graph such that $\spac$ is its vertex set, and two elements $X,Y\in\spac$ form an edge in $\Gamma(\spac)$ if the sum-rank distance between them $\srk(X-Y)$ is exactly~$1$. The sum-rank-metric graph $\Gamma(\spac)$ is sometimes denoted by $\Gamma$ if no confusion arises.
\end{definition}

We highlight two extremal cases of sum-rank-metric graphs. If the number of blocks $t=1$ so that $\bfn = (n)$, $\bfm=(m)$, and so $\spac=\Fq^{n\times m}$ is a rank-metric space, then the graph is called a \textbf{bilinear forms graph}, denoted by $\Bil_q(n,m)$; see \emph{e.g.}~\cite{BCNbookDRG}. On the other hand, if the sizes of all $t$ blocks are $1\times1$, so that $\bfn=\bfm=(1,\dots,1)$, then the graph is called a \textbf{Hamming graph}, denoted by $H(t,q)$, associated with the Hamming space of vectors of length $t$ over a finite field $\Fq$.

In~\cite{AKR2024,AGKP2024,ART2025}, the sum-rank-metric graph has been used to derive bounds on the cardinality of the sum-rank-metric code. The key observation which allows the connection to be made is that the value $A(q,\bfn,\bfm,d)$ is equal to the \textbf{$(d-1)$-independence number} of $\Gamma$, which is the maximal size of a vertex subset such that any two vertices in the subset are at distance at least~$d$ from each other~\cite[Corollary 16]{AKR2024}. The rest of the section is an overview of the upper bounds on the maximal cardinality of a sum-rank-code that arise from this connection to graphs.

\subsubsection{The Ratio-type LP bound}\label{sec:RTboundLP}
The following result was originally presented in~\cite{ACF2019} in the context of regular graphs as an upper bound on the $(d-1)$-independence number, and then further discussed in the context of sum-rank-metric graphs in~\cite{AKR2024}.

\begin{theorem}[Ratio-type bound; c.f. \cite{ACF2019,AKR2024}]\label{thm:hoffman-like} Let $\spac$ be a sum-rank-metric space, and let $A$ be the adjacency matrix of $\Gamma(\spac)$ with eigenvalues \mbox{$\lambda_1\geq\dots\geq\lambda_{|\spac|}$,} where $|\spac|=q^{\sum_{i=1}^t m_in_i}$. Let $p$ be a polynomial with real coefficients of degree at most $d-1$. Define \mbox{$W(p)=\max_{\mathrm{X}\in \spac}\{(p(A))_{\mathrm{X}\mathrm{X}}\}$}, i.e. the maximal element of the diagonal of the matrix $p(A)$, and\\ \mbox{$\lambda(p)=\min_{i=2,\dots,|\spac|}\{p(\lambda_i)\}$.} Then
\begin{equation}\label{eq:hoffman-like}
    A(q,\bfn,\bfm,d) \leq |\spac| \, \frac{W(p)-\lambda(p)}{p(\lambda_1)-\lambda(p)}.
\end{equation}
\end{theorem}

Hence, the application of the bound relies on the choice of the polynomial $p$ which minimizes the right-hand side of Eq.\:\eqref{eq:hoffman-like}. Throughout the paper, the value achieved by choosing an optimal $p$ will be denoted by $\mathrm{RT}(q,\bfn,\bfm,d)$. For $d=3$ and $d=4$, the best choice of the polynomial is known for any regular graph~\cite{ACF2019,KN2022}, including sum-rank-metric graphs. We state the case $d=3$ here due to its relevance for discussions in Section~\ref{sec:nonexistence}.

\begin{theorem}[Ratio-type bound, $d=3$; cf.~\cite{ACF2019}]\label{thm:hoffman-k=2} 
Let $\spac$ be a sum-rank-metric space and let $A$ be the adjacency matrix of $\Gamma(\spac)$ with $r\geq 2$ distinct eigenvalues \mbox{$\theta_0>\theta_1>\dots>\theta_r$}. Let $\theta_i$ be the largest eigenvalue such that $\theta_i\leq -1$. Then $$A(q,\bfn,\bfm,3)\leq |\spac|\frac{\theta_0+\theta_i\theta_{i-1}}{(\theta_0-\theta_i)(\theta_0-\theta_{i-1})}.$$ Moreover, this is the best possible bound that can be obtained by choosing a polynomial via Theorem~\ref{thm:hoffman-like}.    
\end{theorem}

The bound of Theorem~\ref{thm:hoffman-k=2} can be further refined by considering particular families of sum-rank-metric spaces for which the distinct eigenvalues $\theta_i,i\in\{0,\dots,r\},$ can be explicitly expressed in terms of $q,\bfn,\bfm$,~and~$t$. In particular, a special case of sum-rank-metric spaces was considered in~\cite{AKR2024}, for which the Ratio-type bound takes the following form.

\begin{theorem}\cite{AKR2024}\label{thm:rt-bound_family}
    Let $\spac$ be a sum-rank-metric space with \linebreak \mbox{$\bfn=(n,1,\dots,1)$} and \mbox{$\bfm=(m,1,\dots,1)$,} with $m\geq n$ for some integers $m\geq2$ and $n\geq1$. Then $A(q,\bfn,\bfm,3)$ is upper bounded by
    {\scriptsize{
\begin{equation*}\label{eq:alpha2_bound}
\frac{q^{m n+t-1}(q-1)\left((q^m-1)(q^n-1) + (q-1)^2 (t-1) + (q-1)(\varepsilon+1) (\varepsilon-q+1)\right)}{\left(\left(q^m-1\right)
   \left(q^n-1\right)+(q-1)^2 (t-1)+\varepsilon(q-1)+1\right)\left(\left(q^m-1\right) \left(q^n-1\right)+(q-1)^2 (t-2)+\varepsilon(q-1)\right)},
\end{equation*}
}}
where \mbox{$\varepsilon=(t-1)\mod q$.}
\end{theorem}

The cases discussed above arise from introducing limitations on the minimum distance $d$ (namely by taking $d=3$ or $4$) or on the structure of the sum-rank-metric space, like in Theorem~\ref{thm:rt-bound_family}. For the general case, where the polynomial $p$ that minimizes the Ratio-type bound of Theorem~\ref{thm:hoffman-like} is not necessarily known, a linear programming (LP) approach can be used to compute the bound instead. This LP method was introduced in~\cite{F2020} for the so-called partially walk-regular graphs, and takes the minimum distance $d$ and the distinct eigenvalues of the adjacency matrix of the graph as input. We note that sum-rank-metric graphs are partially walk-regular, and their eigenvalues can be efficiently computed, so this LP method is applicable to sum-rank-metric spaces, as discussed in~\cite{AKR2024}.

The implementation and optimization of the polynomial used in the Ratio-type bound is a non-trivial problem that was investigated in \cite{F2020}. The best polynomials for the Ratio-type bound are the socalled minor polynomials. Below we provide a summary, for more details we refer the reader to \cite{F2020}.

\begin{lemma}[cf. \cite{F2020}]\label{l:FiolGen} Let $\spac$ be a sum-rank-metric space, and let\\ \mbox{$\theta_0>\cdots>\theta_r$} be the distinct eigenvalues of the adjacency matrix of $\Gamma(\spac)$. Let $m(\theta_j)$ denote the multiplicity of the eigenvalue $\theta_j$ for $j\in\{0,\dots,r\}$. Finally, let $\mathcal{I}$ be the set of all subsets of $\{1,\dots,r\}$ of cardinality $d-1$. Then 
\begin{equation}\label{eq:MPbound}
    A(q,\bfn,\bfm,d)\leq F(\theta_0,\dots,\theta_n):=\min\limits_{I\in\mathcal{I}} \sum_{j\notin I} m(\theta_j)\prod\limits_{i\in I}\frac{\theta_j-\theta_i}{\theta_0-\theta_i}.    
\end{equation}
\end{lemma}

The alternative formulation of right-hand side of Eq.\:\eqref{eq:MPbound} is as an~LP which follows the notation of Lemma~\ref{l:FiolGen}, with solution $\mathrm{RT}(q,\bfn,\bfm,d)$:
\begin{equation}\label{eq:FiolLP}
\boxed{
\begin{array}{ll@{}l}
\text{minimize}  &\sum_{j=0}^{r} m(\theta_j)x_j &\\
\text{subject to} &f[\theta_0,\dots,\theta_s]=0, \ \quad \ & s=d,\dots,r\\
&x_j\geq0,~& j=1,\dots,r\\
& x_0=1 \\
\end{array}
}
\end{equation}
Here, \mbox{$f[\theta_0,\dots,\theta_s]$} denotes the $s$-th divided difference of Newton interpolation, 
recursively defined by $$f[\theta_i,\dots,\theta_j]=\frac{f[\theta_{i+1},\dots,\theta_j]-f[\theta_i,\dots,\theta_{j-1}]}{\theta_j-\theta_i},$$ where $j>i$, starting with \mbox{$f[\theta_i]=x_i$} for \mbox{$i\in\{0,\dots,r\}$.} The LP formulation\:\eqref{eq:FiolLP}, which is equivalent to Eq.\:\eqref{eq:MPbound}, was first proposed in~\cite{F2020}, where more details can be found.

\subsubsection{The Delsarte LP bound}\label{sec:DelsarteLPsrkcodes}
 
We conclude the overview with another linear programming bound on the cardinality of a sum-rank-metric code. The Delsarte LP approach is a well established bound, originally introduced in the context of the Hamming metric~\cite{DelsarteLP} and subsequently applied to other metrics, e.g.~\cite{D1978,delsarte1975alternating,DIL2020,A1982,S1986,R2010}, including the sum-rank metric~\cite{AGKP2024}. The key step of the method is to define an appropriate association scheme of the metric space under consideration. The invariants of the association scheme then act as an input for the Delsarte LP. In~\cite{AGKP2024}, the so-called sum-rank-metric scheme was defined inspired by the structure of the sum-rank-metric graph, which allows to define the LP as follows.

Let $\spac=\mathbb{F}_q^{n_1\times m_1}\times \cdots \times
\mathbb{F}_q^{n_t\times m_t}$ be a sum-rank-metric space with $m_i\geq n_i$ for $i\in[t]$, and let $\db{n}=0\cup[n]$ for a positive integer $n$. We define a set $\mathcal{R}$ of binary relations on $\spac$ in the following way: for ${\bf j}=(j_1,\dots,j_t)\in\db{n_1}\times\dots\db{n_t}$, two tuples $X=(X_1,\dots,X_t)$ and $Y=(Y_1,\dots,Y_t)\in\spac$ are in relation $R_{\bf j}$ if and only if $\rk(X_i-Y_i)=j_i$ for all $i\in[t]$. In other words, the $t$-tuple $\bf j$ encodes the rank distances between the $t$ elements of the tuples, each in its respective rank space $\F_q^{n_i\times m_i}$. The set $\mathcal{R}$ together with a point set $\spac$ forms the sum-rank-metric scheme; see~\cite{AGKP2024} for more information on the construction and the background on association schemes.

For each $\F_q^{n_i\times m_i}$, $i\in[t]$, we define the matrix $Q^{(i)}$ of size~$(n_i+1)\times(n_i+1)$, which corresponds to the so-called \emph{second eigenmatrix} of a bilinear forms scheme corresponding to the rank space~$\F_q^{n_i\times m_i}$. The first and second eigenmatrices are key invariants that are defined for a symmetric association scheme through a respective Bose-Mesner algebra; see, e.g., \cite{DelsarteLP} or \cite[Chapter 2]{BCNbookDRG} for more details. In practice, using properties of bilinear forms graphs (see e.g.~\cite{BCNbookDRG,AGKP2024}), the entries~$Q^{(i)}$ can be defined recursively using the following relations:
$$Q^{(i)}_{uv}=p_v(\theta_u)=\frac1{c_v} \left( (\theta_u - a_{v-1}) p_{v-1}(\theta_u) - b_{v-2} p_{v-2}(\theta_u)\right),\quad u,v\in \{0,\dots,n_i\},$$
where $p_0(\theta_u)=1$, $p_1(\theta_u)=\theta_u$, and 
\begin{align*}
b_v &:= \frac{q^{2v}(q^{m_i-v}-1)(q^{n_i-v}-1)}{q-1}, \\
c_v &:= \frac{q^{v-1}(q^v-1)}{q-1}, 
\\
a_v &:= b_0 - b_v - c_v.
\end{align*}
From the $t$ matrices $Q^{(1)},\dots,Q^{(t)}$ we obtain the matrix $Q=Q^{(1)}\otimes\cdots\otimes Q^{(t)}$, where $\otimes$ denotes the Kronecker product of matrices. The matrix $Q$ then corresponds to the so-called second eigenmatrix of the sum-rank-metric scheme defined above. This matrix $Q$ acts as the input for the Delsarte LP with variables forming a vector $\bf a$ with entries $a_{\bf j}$ indexed by ${\bf j}=(j_1,\dots,j_t)\in\db{n_1}\times\cdots\times\db{n_t}$:
\begin{equation}\label{eq:LPsrk}
\boxed{
\begin{array}{ll@{}ll}
\text{maximize}  &\sum_{{\bf j}\in \db{n_1} \times\cdots\times \db{n_t}} a_{\bf j} &\\
\text{subject to} &\mathbf{a}Q\geq \mathbf{0},\\
&\mathbf{a}\geq \mathbf{0},\\
&a_{(0,\ldots,0)}=1,\\
&a_{\bf j} = 0, &\quad 0<\sum_{i=1}^t j_i<d.
\end{array}
}
\end{equation}

The vector $\bf a$ is called the \textbf{distribution vector} corresponding to a subset $\mathcal{C}$ of elements of the space (or, alternatively, points of the association scheme). Intuitively, the variable $a_{\bf j}$ with ${\bf j}=(j_1,\dots,j_t)\in\db{n_1}\times\cdots\times\db{n_t}$ corresponds to the number of pairs of matrix tuples $\mathrm{X}=(X_1,\dots,X_t)$ and $\mathrm{Y}=(Y_1,\dots,Y_t)$ in a code $\mathcal{C}$ such that $\rk(X_i-Y_i)=j_i$ for $i\in[t]$, divided by the size of $|\mathcal{C}|$. In other words, ${\bf j}=(j_1,\dots,j_t)$ encodes the partition of the sum-rank distance into $t$ components, each belonging to its own rank-metric space. In particular, the constraint $a_{\bf j}=0$ for $\sum_{i=1}^t j_i<d$ encodes the fact that there are no two codewords in $\mathcal{C}$ at distance less than $d$. The constraint ${\bf a}Q\geq 0$ comes from a key result of~\cite{DelsarteLP}. The goal function, which is the sum of all entries of the vector $\bf a$, is equal to $|\mathcal{C}|$ by the definition of the distribution vector. Hence, the Delsarte LP\:\eqref{eq:LPsrk} above outputs an upper bound on the maximal cardinality of a sum-rank-metric code $\mathcal{C}$ with minimum distance $d$ in $\spac$. Throughout this paper, the optimal solution of the Delsarte LP\:\eqref{eq:LPsrk} is denoted by~$D(q,\bfn,\bfm,d)$.

\subsubsection{The Lov\'asz theta number}\label{sssec:Lovasz}\label{sec;Lovasztheta}

To conclude this section, we mention another prominent bound on the $(d-1)$-independence number of a graph, and hence on the maximal size of a code with minimum distance $d$, arising an algebraic-graph-theoretical approach: the Lov\'asz theta number.

Let $\Gamma$ be a graph on $n$ vertices. We say that a graph $\Gamma^k$ is the \textbf{$k$-th power} of $\Gamma$ when its vertex set is the same as in $\Gamma$, and two vertices in $\Gamma^k$ form an edge if these vertices are at distance at most~$k$ in $\Gamma$. It is easy to see that the $(d-1)$-independence number of $\Gamma$ is exactly the independence number of $\Gamma^{d-1}$.

A very well-known bound on the independence number is the \textbf{Lov\'asz theta number} $\vartheta(\Gamma)$~\cite{Lovasz}. Analogously, a \textbf{Lov\'asz theta-$k$ number} $\vartheta_k(\Gamma)=\vartheta(\Gamma^k)$ can be introduced. The value is estimated using an SDP method as follows~\cite{Lovasz}: 
For a graph $\Gamma$ on $n$ vertices, let $\mathcal{S}_+^n$ denote the set of all $n\times n$ symmetric positive semidefinite matrices. Then $\vartheta(\Gamma)$ is the solution of:
\begin{equation}\label{eq:LovaszSDP}
\boxed{
\begin{array}{ll@{}l}
\text{max}  &\operatorname{tr}(BJ) &\\
\text{subj. to} &\operatorname{tr}B=1, &\\
&b_{ij}=0, &\quad (i,j)\in E,\,i,j\in[n],\\
&B\in\mathcal{S}_+^n, &\\ 
\end{array}
}
\end{equation}
where $J$ is an all-ones matrix of size $n\times n$. Note that $\operatorname{tr}(BJ)$ is the sum of the entries in $B$.

An important connection to the Delsarte LP bound\:\eqref{eq:LPsrk} follows from~\cite{S79}: for graphs derived from symmetric association schemes, including sum-rank-metric graphs, the Lov\'asz theta-$k$ number $\vartheta_k(\Gamma)$ is greater than or equal to the value obtained through Delsarte's LP method~\cite{DelsarteLP}; see also~\cite[Remark~4.1]{AGKP2024}.

\section{Extremal cases of the sum-rank metric: LP bounds equivalence}\label{sec:equivalencesLPbounds}

\subsection{Rank-metric case}

In this section we address a question asked in \cite{AAR2025} and \cite{AGKP2024}.

Recall that a bilinear forms graph $\Bil_q(n,m)$ is equivalent to the sum-rank-metric graph $\Gamma(\F_q^{\bfn\times\bfm})$ with $\bfn=(n)$ and $\bfm=(m)$, i.e., all tuples of matrices are of length $t=1$. It was previously observed in~\cite{AGKP2024} that the output of the Ratio-type bound of Theorem~\ref{thm:hoffman-like} computed using Fiol's LP\:\eqref{eq:FiolLP} coincides with the output of Delsarte's LP bound~\cite{DelsarteLP,D1978} in computational experiments on bilinear forms graphs up to $10^7$ vertices. Here, we prove the equivalence of these two bounds in the general case.

We start with a preliminary result that allows us to compare the two bounds in the rank-metric case.

\begin{lemma}\label{l:DelOpt}\cite[Theorem 5.1.1]{W2022} The optimum of Delsarte's LP for codes with minimum distance $d$ in a rank-metric space~$\Fq^{n \times m}$ is given by $q^{m(n-d+1)}$.
\end{lemma}

Another key result for this section is Lemma~\ref{l:FiolGen}, which is phrased in terms of a sum-rank-metric space and the corresponding graph. However, we note that $F(\theta_0,\dots,\theta_n)$ in Eq.\:\eqref{eq:MPbound} is given in general form and acts as an upper bound on the $(d-1)$-independence number of any graph with given $n+1$ distinct eigenvalues.

We now state the main result of this subsection.

\begin{theorem}\label{thm:equiv_rank}
    In a rank-metric space $\Fq^{n \times m}$, the optimum of Delsarte's LP for codes with minimum distance $d$ is equal to the optimum of the Ratio-type LP\:\eqref{eq:FiolLP}.
\end{theorem}

\begin{proof} We apply Lemma~\ref{l:FiolGen} to the bilinear forms graph $\Bil_q(n,m)$. The eigenvalues $\theta_j$ and their multiplicities $m(\theta_j)$ are given by~\cite{BCNbookDRG}:
\begin{equation}\label{eq:Bil_evals}
\theta_j = \frac{(q^{n-j}-1)(q^m-q^j)-q^j+1}{q-1},\quad m(\theta_j)=\prod_{i=0}^{j-1}\frac{(q^{n-i}-1)(q^m-q^i)}{q^{i+1}-1}
\end{equation}
for $j\in\{0,\dots,n\}$. In particular, $F(\theta_0,\dots,\theta_n)$ is upper-bounded by
\begin{equation}\label{eq:MPref}
    F(\theta_0,\dots,\theta_n)\leq\sum\limits_{j=0}^{n-d+1} \left( \prod_{i=0}^{j-1}\frac{(q^{n-i}-1)(q^m-q^i)}{q^{i+1}-1} \cdot \prod\limits_{i=n-d+2}^n\frac{\theta_j-\theta_i}{\theta_0-\theta_i}\right),
\end{equation}
which is derived from Eq.\:\eqref{eq:MPbound} by setting $I=\{n-d+2,\dots,n\}$. Next, it is straightforward to show that
$$\frac{\theta_j-\theta_i}{\theta_0-\theta_i} = \frac{q^{i-j}-1}{q^i-1},$$ meaning that Eq.\:\eqref{eq:MPref} can be rewritten as
\begin{equation}\label{eq:MPref2}
    F(\theta_0,\dots,\theta_n)\leq\sum\limits_{j=0}^{n-d+1} \left( \prod_{i=0}^{j-1}\frac{(q^{n-i}-1)(q^m-q^i)}{q^{i+1}-1} \cdot \prod\limits_{i=n-d+2}^n\frac{q^{i-j}-1}{q^i-1}\right).
\end{equation}

We consider the product
{\footnotesize{
\begin{align*}
    &\prod_{i=0}^{j-1}(q^{n-i}-1) \cdot \prod\limits_{i=n-d+2}^n\frac{q^{i-j}-1}{q^i-1} = \\
    &= \frac{(q^n-1)(q^{n-1}-1)\cdots(q^{n-j+1}-1)\cdot (q^{n-j}-1)(q^{n-1-j}-1)\cdots(q^{n-d+2-j}-1)}{(q^n-1)(q^{n-1}-1)\cdots(q^{n-d+2}-1)} =\\
    &= (q^{n-d+1}-1)(q^{n-d}-1)\cdots(q^{n-d+2-j}-1) = \prod\limits_{i=0}^{j-1} (q^{n-d+1-i}-1).
\end{align*}
}}

Then, Eq.\:\eqref{eq:MPref2} is equivalent to
\begin{equation}\label{eq:MPref3}
    F(\theta_0,\dots,\theta_n)\leq\sum\limits_{j=0}^{n-d+1} \prod_{i=0}^{j-1}\frac{(q^{n-d+1-i}-1)(q^m-q^i)}{q^{i+1}-1}.    
\end{equation}
Note that the product in Eq.\:\eqref{eq:MPref3} is equal to the multiplicity of $j$-th largest distinct eigenvalue of the bilinear forms graph $\Bil_q(n-d+1,m)$. Hence, the right-hand side of Eq.\:\eqref{eq:MPref3} is exactly the number of vertices of $\Bil_q(n-d+1,m)$, which is equal to $q^{m(n-d+1)}$:
\begin{equation}\label{eq:MPfinal}
    F(\theta_0,\dots,\theta_n)\leq q^{m(n-d+1)}.
\end{equation}

Now, suppose $\mathcal{C}$ is a rank-metric code with minimum distance $d$ that corresponds to an optimal solution to Delsarte's LP. From Lemma~\ref{l:DelOpt}, we have $|\mathcal{C}|=q^{m(n-d+1)}$. Hence, Eq.\:\eqref{eq:MPbound} together with Eq.\:\eqref{eq:MPfinal} imply
$$q^{m(n-d+1)}=|\mathcal{C}|\leq F(\theta_0,\dots,\theta_n)\leq q^{m(n-d+1)},$$
meaning that $$F(\theta_0,\dots,\theta_n)= q^{m(n-d+1)}.$$
This means that the minimum in Eq.\:\eqref{eq:MPbound} is achieved by choosing $I=\{n-d+2,\dots,n\}$, and so in the rank metric both the optimum of Delsarte's LP and the optimum of the Ratio-type LP are equal to $q^{m(n-d+1)}$ and to each other.
\end{proof}

In particular, the proof of Theorem~\ref{thm:equiv_rank} implies that, for a bilinear forms graph $\Bil_q(n,m)$, the $(d-1)$-minor polynomial $f_{d-1}(x)$ defined in~\cite{F2020} takes the form
$$f_{d-1}(x)=\prod\limits_{i=n-d+2}^n\frac{x-\theta_i}{\theta_0-\theta_i}.$$

\subsection{Hamming case}

Recall that $H(t,q)$ denotes the Hamming graph, which is a particular case of the sum-rank-metric graph $\Gamma(\F_q^{\bf 1\times 1})$, where ${\bf1}=(1,\dots,1)$ is a $t$-tuple of ones. In other words, the vertex set of the graph $H(t,q)$ consists of all vectors of length $t$ over a finite field $\F_q$, and two vertices are adjacent if the Hamming distance between the two vectors is exactly $1$.

In \cite{F2020}, the equivalence of the Ratio-type LP bound and Delsarte's LP bound has been briefly discussed over a series of examples. In particular for the Hamming case, the Hamming graph $H(7,2)$ was considered in detail. It was shown that, for any choice of minimum distance $d$, the two LP bounds are equivalent. In this section, we build on this discussion and present a comparison of the two LP bounds, proving some new equivalences of the bounds under certain conditions.

We start with several preliminary lemmas.

\begin{lemma}\label{l:Ham_Del_leq_RT}
    In the Hamming metric space of vectors of length $t$ over a finite field $\F_q$, the optimum of Delsarte's LP for codes with minimum distance $d$ is at most the optimum of the Ratio-type LP bound.
\end{lemma}

\begin{proof} This is observed by comparing the two bounds with the Lov\'asz theta-$(d-1)$ number $\vartheta_{d-1}(\Gamma)$, which is defined on a graph $\Gamma$ and can be seen as an upper bound on the maximal size of a code with minimum distance $d$; see Section~\ref{sssec:Lovasz}, as well as~\cite{Lovasz} and~\cite[Remark 23.1]{AKR2024} for more details.

Firstly, from~\cite{S79} we observe that, for any choice of $d$, the optimum of Delsarte's LP is less than or equal to the Lov\'asz theta-$(d-1)$ number $\vartheta_{d-1}(\Gamma)$ whenever the point set of the association scheme $\mathcal{A}$ on which Delsarte's LP is defined is exactly the vertex set of $\Gamma$, and the edge set of $\Gamma$ can be represented as a union of relation classes of $\mathcal{A}$. Since the Hamming scheme $\mathcal{A}$ is defined from distance relations of the Hamming graph $\Gamma$, this is trivially true in the case considered here. 

Secondly, in~\cite[Section 6]{F2020}, it is shown that, if the underlying graph is \emph{walk-regular}, the Lov\'asz theta-$(d-1)$ number is at most the optimum of the Ratio-type LP on codes with minimum distance $d$. A graph is \emph{walk-regular} if, for any positive integer $k$, the number of closed walks of length $k$ does not depend on the choice of the starting vertex. It is a well-known fact that Hamming graphs are walk-regular; for instance, one may recall that Hamming graphs are a special case of sum-rank-metric graphs, which are always walk-regular~\cite[Proposition 11]{AKR2024}.

By combining the two facts, the statement of the lemma follows.
\end{proof}

\begin{lemma}\label{l:Ham_RT_leq_S}
    In the Hamming metric space of vectors of length $t$ over a finite field $\F_q$, the optimum of the Ratio-type LP for codes with minimum distance $d$ is at most the Singleton bound, which equals $q^{t-d+1}$.
\end{lemma}

\begin{proof}
    The method of the proof is similar to one presented in Theorem~\ref{thm:equiv_rank} for the rank-metric case. Namely, we apply Lemma~\ref{l:FiolGen} to Hamming graphs.

    The eigenvalues $\theta_j$ and their multiplicities $m(\theta_j)$ are given by~\cite[Theorem 9.2.1]{BCNbookDRG}:
\begin{equation*}
\theta_j = q(t-j)-t,\quad m(\theta_j)={t \choose j}(q-1)^j
\end{equation*}
for $j\in\{0,\dots,t\}$. In particular, by setting $I=\{t-d+2,\dots,t\}$, we observe from Eq.\:\eqref{eq:MPbound} that $F(\theta_0,\dots,\theta_t)$ is upper-bounded by
\begin{align*}
    F(\theta_0,\dots,\theta_t) &\leq\sum\limits_{j=0}^{t-d+1} \left( {t \choose j}(q-1)^j \cdot \prod\limits_{i=t-d+2}^t\frac{q(t-j)-t-q(t-i)+t}{q(t-0)-t-q(t-i)+t}\right) \\
    &= \sum\limits_{j=0}^{t-d+1} \left( {t \choose j}(q-1)^j \cdot \prod\limits_{i=t-d+2}^t\frac{i-j}{i}\right) \\
    &= \sum\limits_{j=0}^{t-d+1} \left( {t \choose j}(q-1)^j \cdot \frac{(t-d+1)!\cdot(t-j)!}{t!\cdot(t-d+1-j)!} \right) \\
    &=\sum\limits_{j=0}^{t-d+1} {t-d+1 \choose j}(q-1)^j.
\end{align*}

Note that the term ${t-d+1 \choose j}(q-1)^j$ is equal to the multiplicity of $j$-th largest distinct eigenvalue of the Hamming graph $H(t-d+1,q)$. Hence, $\sum\limits_{j=0}^{t-d+1} {t-d+1 \choose j}(q-1)^j$ is exactly the number of vectors in $H(t-d+1,q)$, which is equal to $q^{t-d+1}$. Thus,
$$F(\theta_0,\dots,\theta_n)\leq q^{t-d+1},$$
which is exactly the inequality we aimed to show.
\end{proof}

Combining Lemmas~\ref{l:Ham_Del_leq_RT} and~\ref{l:Ham_RT_leq_S}, we  derive the following result.

\begin{theorem}\label{thm:HammingBigIneq}
    For a Hamming space~$\F_q^{\bf 1\times1}$ on vectors of length~$t$ over a finite field~$\Fq$, we have
$$D(q,{\bf 1},{\bf 1},d)\leq \vartheta_{d-1}(H(t,q))\leq \mathrm{RT}(q,{\bf 1},{\bf 1},d)\leq q^{t-d+1},$$
where ${\bf 1}=(1,\dots,1)$ is a $t$-tuple of ones, while $D(q,{\bf 1},{\bf 1},d)$ and $\mathrm{RT}(q,{\bf 1},{\bf 1},d)$ denote the optima of Delsarte's LP and Ratio-type LP for codes with minimum distance $d$ in the Hamming space $\F_q^{\bf 1\times1}$, respectively.
\end{theorem}

Theorem~\ref{thm:HammingBigIneq} implies that the question of the equivalence between Delsarte's LP and the Ratio-type LP bounds is closely related to the question of finding optimum of Delsarte's LP in a Hamming space. In regards to the latter, the following result is known:

\begin{lemma}\label{l:DelOptHam}\cite{DelsarteLP}, \cite[Theorem 3.1(b)]{SW2025}
    In the Hamming metric space of vectors of length $t$ over a finite field $\F_q$, the optimum of Delsarte's LP for codes with minimum distance $d$ is equal to $q^{t-d+1}$ if $q\geq\max\{d,t-d+2\}$.
\end{lemma}

The above lemma has the following immediate consequence:

\begin{corollary}
    In the Hamming metric space of vectors of length $t$ over a finite field $\F_q$, the optimum of Delsarte's LP for codes with minimum distance $d$ is equal to the optimum of the Ratio-type LP if $q\geq\max\{d,t-d+2\}$.
\end{corollary}

No general closed formula is known for the optimum of Delsarte’s LP in the Hamming scheme. Lemma~\ref{l:DelOptHam} gives one important family of parameters for which we know the optimum. Naturally, any progress on the question of deriving closed expressions for the optimum of Delsarte’s LP gives immediate results for both the Ratio-type LP as well as the value of the Lov\'asz theta number in powers of Hamming graphs.

We note that the equality between the optima of Delsarte's and the Ratio-type LPs does not necessarily imply the optima reach the Singleton bound~$q^{t-d+1}$. The following theorem gives a family of examples that demonstrate this.

\begin{theorem}\label{thm:Ham_q=2_d=t-1}
    In the Hamming metric space of vectors of length $t\geq2$ over the binary field $\F_2$, the optima of Delsarte's and the Ratio-type LPs for codes with minimum distance $t-1$ are equal to 
    $$c(t)=\begin{cases}
        2+\frac2{t-1}, & t \text{ is even,}\\
        2+\frac2{t-2}, & t \text{ is odd.}
    \end{cases}$$
\end{theorem}

Theorem~\ref{thm:Ham_q=2_d=t-1} implies that both LP bounds estimate the size of the code in question to be $2$ for any $t\geq4$, while the Singleton bound returns $4$.

\begin{proof}
    We first use Lemma~\ref{l:FiolGen} to show that the optimum of the Ratio-type LP is upper-bounded by $c(t)$. Then, we will show that the optimum of Delsarte's LP is lower-bounded by $c(t)$, which forces the equality.

    We prove the case when $t$ is even; the case when $t$ is odd can be shown analogously. By taking Eq.\:\eqref{eq:MPbound} and setting $I=\{2,\dots,t-1\}$ we obtain
    \begin{align*}
    F(\theta_0,\dots,\theta_t) &\leq \sum\limits_{j=0,1,t} \left( {t \choose j}(q-1)^j \cdot \prod\limits_{i=2}^{t-1}\frac{i-j}{i}\right) \\
    &= \sum\limits_{j=0,1,t} \left( {t \choose j}\cdot \prod\limits_{i=2}^{t-1}\frac{i-j}{i}\right)  \\
    &={t \choose 0}\cdot \prod\limits_{i=2}^{t-1}\frac{i-0}{i} +{t \choose 1}\cdot \prod\limits_{i=2}^{t-1}\frac{i-1}{i}+{t \choose t}\cdot \prod\limits_{i=2}^{t-1}\frac{i-t}{i} \\
    &=1+t\frac{(t-2)!}{(t-1)!}+(-1)^{t-2}\frac{(t-2)!}{(t-1)!} = 2+\frac2{t-1},\label{eq:example_RTub}
\end{align*}
where we use that $q=2$ and $t$ is even.

For the case when $t$ is odd, we obtain the bound on $F(\theta_0,\dots,\theta_t)$ by taking $I=\{2,\dots,t-2,t\}$.

To prove $c(t)$ is an upper bound on the optimum of Delsarte's LP, we use its representation through Krawtchouk polynomials; see e.g.~\cite{BD2021} for an overview. Namely, the optimum of Delsarte's LP is given by $f(0)$. Here, $f(z)$ is a real polynomial defined by $$f(z) = \sum_{i=0}^t f_i K_i^{(t,q)}(z),$$ where $K_i^{(n,q)}(z)$ are the Krawtchouk polynomials, and the coefficients $f_i$ are such that $f_0=1$, $f_i\geq 0$ for $i\in[t]$, while $f(0)>0$ and $f(i)\leq 0$ for $i\in\{d,\dots,t\}$. In our case $d=t-1$, and the latter condition reduces to $f(t-1)\leq0$ and $f(t)\leq0$. Additionally, for the binary Krawtchouk polynomials we have
$$K_i^{(t,2)}(0)={t \choose i}; \quad K_i^{(t,2)}(t)=(-1)^i{t \choose i}; \quad K_i^{(t,2)}(t-1)=(-1)^i\frac{t-2i}t{t \choose i},$$
which means
\begin{align*}
    f(0) &= \sum\limits_{i=0}^t f_i {t\choose i}, \\
    f(t-1) &= \sum\limits_{i=0}^t (-1)^i\frac{t-2i}t f_i {t\choose i}\leq 0, \\
    f(t) &= \sum\limits_{i=0}^t (-1)^i f_i {t\choose i} \leq 0.
\end{align*}
Let $x_i:= f_i {t\choose i}, i\in[t]$ be the variables of Delsarte's LP for this case (recall that $f_0{t\choose 0}=1$):
\begin{equation*}
    \boxed{
\begin{array}{ll@{}l}
\text{minimize}  &\sum\limits_{i=1}^t x_i\\
\text{subject to} &\sum\limits_{i=1}^t (-1)^{i+1}\frac{t-2i}t x_i\geq 1,&\\
&\sum\limits_{i=1}^t (-1)^{i+1} x_i\geq 1,&\\
&x_i\geq 0,&i\in[t].\\
\end{array}
}
\end{equation*}
The dual of this LP is
\begin{equation*}
    \boxed{
\begin{array}{ll@{}l}
\text{maximize}  &y_1+y_2\\
\text{subject to} &(-1)^{i+1}y_1+(-1)^{i+1}\frac{t-2i}t y_2\leq 1,&\;i\in[t],\\
&y_i\geq 0,&\;i\in\{1,2\}.\\
\end{array}
}
\end{equation*}
Any feasible solution of the dual LP is a lower bound on the optimum of the primal LP, i.e. the optimum Delsarte's LP $f(0)$ minus 1. As a feasible solution, we take $y_1 = \frac1{t-1}$ and $y_2=\frac{t}{t-1}$. Indeed, the constraints are satisfied: clearly, $y_1$ and $y_2$ are both at least $0$ for $t\geq2$, and
$$
    (-1)^{i+1}\left(\frac1{t-1}+\frac{t-2i}{t-1}\right) = (-1)^{i+1}\frac{t-2i+1}{t-1} = (-1)^{i+1} \left(1-2\frac{i-1}{t-1}\right)\leq 1
$$
for any $i\in[t]$.

Then we have 
\begin{equation}\label{eq:example_Dellb}
    f(0)\geq 1+y_1+y_2=1+\frac1{t-1}+\frac{t}{t-1}=2+\frac2{t-1}.
\end{equation}

If $t$ is even, we combine the lower bound on the optimum of Delsarte's LP\:\eqref{eq:example_Dellb} with the upper bound on the optimum of the Ratio-type LP \eqref{eq:FiolLP} to obtain the equality.

If $t$ is odd, we can further maximize $y_1+y_2$ by taking $y_1=0$ and $y_2=\frac{t}{t-2}.$ Then $y_1\geq0$, $y_2\geq 0$, and the final constraint becomes
$$(-1)^{i+1}\frac{t-2i}{t-2} = (-1)^{i+1} \left(1-2\frac{i-1}{t-2}\right)\leq 1.$$
For any odd $i\in[t]$, this becomes $$1-2\frac{i-1}{t-2}\leq 1\Leftrightarrow i\geq 1,$$ which is true by definition of $i$.
For any even $i\in[t]$, the constraint is equivalent to $$2\frac{i-1}{t-2}-1\leq 1\Leftrightarrow i\leq t-1,$$ which is also true since an even $i$ cannot be equal to an odd $t$.

Hence, in case $t$ is odd, the optimum of Delsarte's LP is lower-bounded by $$f(0)\geq 1+y_1+y_2=1+0+\frac{t}{t-2}=2+\frac2{t-2},$$ and, combined with the upper bound on the optimum of the Ratio-type LP\:\eqref{eq:FiolLP}, we obtain the desired equality.
\end{proof}

We conclude this section with Table~\ref{tab:Ham_Del_RT}, which lists some sporadic parameter values $t,q,d$ under which the optima of Delsarte's and the Ratio-type bounds coincide, but are strictly smaller than $q^{t-d+1}$. Note that such cases usually occur when the minimum distance $d$ is either very small or very close to $t$, which indicates that an approach similar to Theorem~\ref{thm:Ham_q=2_d=t-1} could be helpful to prove equivalence in other series of examples.

\section{A Schrijver-type SDP bound for sum-rank-metric codes}\label{sec:SchrijverSDPbound}

In this section, we consider an SDP approach for estimating the maximal size of a sum-rank-metric code with given minimum distance, which shows improvements on the best known bounds for sum-rank codes (the Delsarte LP bound~\cite{AGKP2024}) and the Lov\'asz theta bound~\eqref{eq:LovaszSDP} for several parameter regimes. Indeed, this approach can be seen as a refinement of Delsarte's LP bound from Section \ref{sec:DelsarteLPsrkcodes}, in the sense that the SDP considers triples of codewords and the distance relations between them, whereas Delsarte's LP bound focuses on the distribution of pairwise distances between the codewords.

This SDP approach was first introduced by Schrijver~\cite{S2005} (hence the name) for binary Hamming codes, and has been later reformulated in a more general form and developed to take quadruples of codewords~\cite{GMS2012} into account. Similar SDP bounds have also been applied to non-binary Hamming codes~\cite{LPS16,GST06}, binary constant weight Hamming codes~\cite{Polak19_B4}, and Lee codes~\cite{Polak19_Lee}. The goal of this section is to show an analogue of this SDP approach to the sum-rank metric, which has not been developed before in this setting.

\subsection{The three-point SDP-bound}\label{ssec:sdp.sdp}

Since the sum-rank metric is translation invariant, we may assume that an optimal code contains $\bf 0$, the zero codeword. We will use this to define a three-point SDP in which the variables encode information on triples of codewords containing $\bf 0$.

Let $X=(x_{u,v})$ be a square matrix of size $|\spac|$, with its rows and columns indexed by the elements of $\spac$. We also assume the codeword ${\bf 0}$ indexes the first row and the first column of $X$. We consider the following SDP:
\begin{equation}\label{eq:julia}
\boxed{
\begin{array}{ll@{}l}
\text{maximize}  &\tr(X) &\\
\text{subject to} &x_{{\bf 0},{\bf 0}}=1,~&\\
&x_{{\bf 0},u}=x_{u,u},~& \forall u\in\spac,\\
&x_{u,v}=0,~& \text{ if }\srk(\{{\bf 0},u,v\})<d,\\
&x_{u,v}=x_{v,u},~& \forall u,v\in\spac,\\
&x_{u,v}\geq0,~& \forall u,v\in\spac,\\
&X\text{ is positive semidefinite.}
\end{array}
}
\end{equation}

The SDP \eqref{eq:julia} is Schrijver-type in nature, since it strengthens the Delsarte framework, which is based on pairs of codewords, by imposing constraints coming from triples of codewords containing the zero word. Rather than introducing both semidefinite constraints and all accompanying linear constraints from Schrijver’s approach~\cite{S2005}, we focus on the single semidefinite constraint above. This already yields a substantial three-point strengthening of the pairwise Delsarte approach. Define 
$$
\mathrm{SDP}(q,\bfn,\bfm,d) := \text{ the optimum value of the SDP~\eqref{eq:julia}}.
$$
\begin{proposition}
It holds that $ A(q,\bfn,\bfm,d) \leq \mathrm{SDP}(q,\bfn,\bfm,d)$.
\end{proposition}
\begin{proof}
Let $C$ be a sum-rank-metric code in $\F^{\bfn\times\bfm}_q$ with minimum distance $d$ of size exactly $A(q,\bfn,\bfm,d)$ and containing the codeword $\bf 0$. Define the $0,1$ matrix~$X$ of size $|\spac|$ by 
$x_{u,v} =1$ if $\{{\bf 0}, u, v\}\subseteq C$ and $x_{u,v}=0$ otherwise. Then $X$ satisfies the constraints in~\eqref{eq:julia}. To see that $X$ is positive semidefinite, note $X$ is a rank one matrix $X=zz^\top$, where $z \in \R^{|\spac|}$ is the indicator vector of $C$. Moreover, the objective value is  $\tr(X)=|C|=A(q,\bfn,\bfm,d)$. So $ A(q,\bfn,\bfm,d) \leq \mathrm{SDP}(q,\bfn,\bfm,d)$.
\end{proof}

The bound can be compared to Delsarte's bound~\cite{DelsarteLP}. 
To compare the bounds, the analytic formulation of Delsarte's bound is used, which is equivalent to the LP formulation\:\eqref{eq:LPsrk} discussed in this chapter~\cite{S79}. Here, $B=(b_{u,v})$ is a square matrix with its rows and columns indexed by the elements of $\spac$. Let $N:=|\spac|$. Delsarte's bound $\Del$ on the size of a sum-rank-metric code with minimum distance $d$ is then given by (\emph{cf.}~\cite[Eq.\:(22)]{S79}):

\begin{equation}\label{eq:DelLP_sdp}
\boxed{
\small
\begin{array}{ll@{}l}
\text{maximize}  &\sum\limits_{u,v\in \spac} b_{u,v} &\\
\text{subject to} &b_{u,v}\geq0,~&\forall u,v\in\spac,\\
&b_{u,v} = b_{v,u},~&\forall u,v\in\spac,\\
&b_{u,v}=0,~& \text{ if }\srk(u-v)<d\text{ with }u\neq v,\\
&\tr(B)=1,~&B\in\mathbb{R}^{N\times N},\\
&B\text{ is positive semidefinite.}
\end{array}
}
\end{equation}

Here, $\tr(B)$ denotes the trace of the matrix $B$, \emph{i.e.} $\tr(B):=\sum_{u\in\spac}b_{u,u}$.

\begin{proposition}\label{prop:comparisonwithdelsarte}
It holds that $\mathrm{SDP}(q,\bfn,\bfm,d)\leq \Del$.
\end{proposition}

\begin{proof}
Let $X=(x_{u,v})_{u,v\in\spac}$ be an optimal solution of the SDP~\eqref{eq:julia}. We show that from $X$ one can construct a feasible solution of Delsarte's SDP~\eqref{eq:DelLP_sdp} whose objective value is at least $\tr(X)$. Write $X$ in block form as
\[
X=
\begin{pmatrix}
1 & y^\top\\
y & Y
\end{pmatrix},
\]
where the first row and column correspond to ${\bf 0}\in\spac$, the vector $y$ is indexed by $\spac\setminus\{{\bf 0}\}$, and $Y$ is a square matrix indexed by $\spac\setminus\{{\bf 0}\}$. By the constraint $x_{{\bf 0},u}=x_{u,u}$ for all $u\in\spac$, we have $y_u=Y_{u,u}$ for all $u\in\spac\setminus\{{\bf 0}\}$. Hence, ${\bf 1}^\top y=\tr(Y)$, where ${\bf 1}$ denotes the all-ones vector.

Now define $B:=\frac{1}{\tr(X)}\,X$. Then $B$ is feasible for~\eqref{eq:DelLP_sdp}. To see this, first, since $X$ is symmetric, entrywise nonnegative, and positive semidefinite, the same is true for $B$. Moreover, $\tr(B)=\frac{\tr(X)}{\tr(X)}=1$. To see the distance constraints, let $u,v\in\spac$ with $u\neq v$ and $\srk(u-v)<d$. Then the set $\{{\bf 0},u,v\}$ has minimum sum-rank distance smaller than $d$, and so the constraint in~\eqref{eq:julia} gives $x_{u,v}=0$. Therefore, $B$ is feasible for Delsarte's SDP.

We now compare the objective values. Since $X$ is positive semidefinite and its $( {\bf 0},{\bf 0})$-entry is equal to $1$, by taking a Schur complement we have $Y-yy^\top \succeq 0$. Therefore,
\[
{\bf 1}^\top Y{\bf 1}\geq {\bf 1}^\top yy^\top {\bf 1}=(\tr(Y))^2.
\]
Using this, we obtain
\begin{align*}
\sum_{u,v\in\spac} x_{u,v}
&=
1+2{\bf 1}^\top y+{\bf 1}^\top Y{\bf 1}\\
&\geq 1+2\tr(Y)+(\tr(Y))^2=(1+\tr(Y))^2\\&=\tr(X)^2.
\end{align*}
Hence the objective value of $B$ in~\eqref{eq:DelLP_sdp} satisfies $$\sum_{u,v\in\spac} B_{u,v}= \tfrac{\sum_{u,v\in\spac} x_{u,v}}{\tr(X)}
\geq \tr(X)=\mathrm{SDP}(q,\bfn,\bfm,d).$$ 
So $\mathrm{SDP}(q,\bfn,\bfm,d)\leq \Del$, as desired.
\end{proof}
The previous two propositions show that \eqref{eq:julia} yields a valid SDP-upper bound on $A(q,\bfn,\bfm,d)$ and is always at least as strong as Delsarte’s bound:
$$
A(q,\bfn,\bfm,d) \leq \mathrm{SDP}(q,\bfn,\bfm,d)\leq \Del
$$
So \eqref{eq:julia} should be viewed as a three-point (Schrijver-type) strengthening of the pairwise Delsarte approach for the sum-rank metric.

We numerically computed our bounds for sum-rank-metric graphs using the Jordan symmetry reduction package in Julia~\cite{BK2022_jordan}. Table~\ref{tab:sdp} lists some sum-rank-metric graphs on up to $2000$ vertices for which the output of the SDP\:\eqref{eq:julia} improves on Delsarte's LP bound from\:\eqref{eq:LPsrk}. Note that in all of these cases, Delsarte's LP bound outperforms the Ratio-type bound of Theorem~\ref{thm:hoffman-like} and the coding bounds of Theorems~\ref{thm:induced} and \ref{thm:non-induced}, meaning that the new computationally obtained SDP bound\:\eqref{eq:julia} also outperforms them.

\section{Non-existence of MSRD and perfect codes}\label{sec:nonexistence}

In this section, we study the existence of extremal sum-rank-metric codes, with particular focus on \emph{maximum sum-rank distance} (MSRD) codes and \emph{perfect codes}. While the previous sections were devoted to deriving upper bounds on $A(q,{\bf n},{\bf m},d)$, we now use these bounds to obtain structural results and, in particular, non-existence statements.

The main idea of this section is to exploit the hierarchy of bounds developed throughout the paper—namely, the Delsarte linear programming bound, the Ratio-type bound, and the semidefinite programming (SDP) bound—to investigate the feasibility of such extremal codes. Whenever one of these bounds is strictly smaller than the value predicted by a putative optimal construction (for instance, the Singleton bound), this discrepancy implies the non-existence of codes achieving that bound.

We first consider MSRD codes and identify parameter regimes where the Delsarte LP bound improves upon the classical bounds, thereby ruling out the existence of Singleton-optimal codes. In several cases, computational evidence shows that the Delsarte bound provides the only obstruction among the known bounds, highlighting its strength in the sum-rank setting.

We then turn to perfect codes. By combining Sphere-Packing arguments with estimates on the volume of balls in the sum-rank metric, we derive explicit conditions under which perfect codes cannot exist. These results demonstrate how analytic bounds can be used not only to estimate the size of optimal codes, but also to exclude their existence entirely.

Overall, this section illustrates how the interplay between combinatorial, spectral, and optimization-based bounds yields new insights into the structure and limitations of sum-rank-metric codes.

\subsection{MSRD codes}

In this section, we consider the sum-rank-metric space $\spac$ with $\bfm=\bfn=(2,\dots,2)$ for arbitrary $t\geq2$ and a prime power $q\geq2$. In this case, the graph $\Gamma(\spac)$ is the Cartesian product of $t$ copies of the bilinear forms graph $\Gamma(\F_q^{2\times 2})$. In Theorem~\ref{thm:RT22} we present the form the Ratio-type bound\:\eqref{eq:alpha2_bound} takes in case $d=3$.

\begin{theorem}\label{thm:RT22} Let $\spac$ be the sum-rank-metric space with $\bfm=\bfn=(2,\dots,2)$ for $t\geq2$ and a prime power $q\geq2$. Then the size of the code with minimum distance $d=3$ is upper-bounded by
$$A(\bfn,\bfm,3)\leq q^{4t}\frac{t(q^2-1)(q+1)-(1+\epsilon)(q^2-1-\epsilon)}{(t(q^2-1)(q+1)+1+\epsilon)(t(q^2-1)(q+1)+1+\epsilon-q^2)},$$
where $\epsilon=(tq+t-1)\mod q^2$.
\end{theorem}

\begin{proof}
    We first observe that the graph $\Gamma(\F_q^{2\times 2})$ has exactly three distinct eigenvalues (see, e.g., \cite[ Corollary 8.4.2]{BCNbookDRG}):
    \begin{gather*}
    \frac{q^4-2q^2+1}{q-1}=(q^2-1)(q+1),\quad \frac{q^3-2q^2+1}{q-1}=q^2-q-1,\\ \frac{q^2-2q^2+1}{q-1}=-(q+1).    
    \end{gather*}
    Since the eigenvalues of a Cartesian product are all possible sums of the eigenvalues of its Cartesian factors, each eigenvalue of $\Gamma(\spac)$ takes the form
    $$a(q^2-1)(q+1)+b(q^2-q-1)-(t-a+b)(q+1)=aq^2(q+1)+bq^2-t(q+1),$$
    where $a,b$ are non-negative integers with $a+b\leq t$.

    Our goal is to identify the two eigenvalues closest to $-1$, in particular the largest eigenvalue that does not exceed $-1$. For this we fix $a\leq\frac1{q^2}\left(t+\frac1{q+1}\right)$, then $\lambda\leq -1$ for some eigenvalue $\lambda$ of $\Gamma(\spac)$ is equivalent to
    $$b\leq\frac{(t-aq^2)(q+1)-1}{q^2}.$$
    Note that the value on the right-hand side is non-negative due to the conditions imposed on $a$. Let $\epsilon = (tq+t-1)\mod q^2$. Then the two eigenvalues closest to $-1$ are
    $$\theta_i=(aq^2-t)(q+1)+\floor{\frac{(t-aq^2)(q+1)-1}{q^2}}q^2 = -1 -\epsilon\quad\text{and}\quad \theta_{i-1}=q^2-1 -\epsilon.$$
    Note that the number of vertices in $\Gamma(\spac)$ is equal to $q^{4t}$, and the valency of the graph is $\delta=\frac{t(q^2-1)^2}{q-1}=t(q^2-1)(q+1)$. From these observations the claim of the theorem follows.
\end{proof}

Now, we can compare the Ratio-type bound of Theorem~\ref{thm:RT22} with the Singleton bound, which in the case of this section takes the form $A(q,\bfn,\bfm,3)\leq q^{4(t-1)}$ for $d=3$.

\begin{theorem}
    Let $\spac$ be the sum-rank-metric space with $\bfm=\bfn=(2,\dots,2)$ for $t\geq2$ and a prime power $q\geq2$. Then the Ratio-type bound does not perform worse than the Singleton bound for a code with minimum distance $d=3$. Moreover, if additionally $t\geq q-1$, then the Ratio-type bound strictly outperforms the Singleton bound.
\end{theorem}

\begin{proof}
    For the Ratio-type bound to perform at least as good as the Singleton bound in the described case, the following inequality needs to hold:
    $$q^{4t}\frac{t(q^2-1)(q+1)-(1+\epsilon)(q^2-1-\epsilon)}{(t(q^2-1)(q+1)+1+\epsilon)(t(q^2-1)(q+1)+1+\epsilon-q^2)}\leq q^{4(t-1)},$$
where $\epsilon=(tq+t-1)\mod q^2$. This can be rewritten as
{\footnotesize{
\begin{equation}\label{eq:22RTvsS}
    \frac{q^{4t-2}(q-1)(q+1) (tq+t-1 - \epsilon) ( (q-1)(q+1)^2 t-(q^2-1-\epsilon)(q^2+1))}{(1+    \epsilon + (q-1) (q+1)^2 t) (\epsilon + (q^2-1) (qt+t-1))} \geq 0.
\end{equation}
}}
Since $q\geq 2$ and $t\geq 2$, the denominator is clearly positive. Since $\epsilon = (tq+t-1)\mod q^2$, we have $tq+t-1 - \epsilon\geq0$ with equality reached if and only if $tq+t-1 < q^2 \Leftrightarrow t < \frac{q^2+1}{q+1}=q-1+\frac2{q+1}$, which for integer values $t,q\geq2$ is equivalent to $t< q-1$. So, the inequality\:\eqref{eq:22RTvsS} reduces to
$$(q-1)(q+1)^2 t-(q^2-1-\epsilon)(q^2+1)\geq0\Leftrightarrow t\geq \frac{(q^2-1-\epsilon)(q^2+1)}{(q-1)(q+1)^2}.$$
We consider three cases: $t\geq q$, $t=q-1$, and $t\leq q-2$. In case $t\geq q$, then we have
$$
\frac{(q-1)(q+1)^2t}{q^2+1}\geq \frac{(q-1)(q+1)^2q}{q^2+1}> q^2-1\geq q^2-1-\epsilon,
$$
where the second inequality holds for all $q\geq 2$.

In case $t\leq q-2$, the inequality\:\eqref{eq:22RTvsS} reduces to $0\geq 0$ due to $tq+t-1=\epsilon$.

In case $t= q-1$, we have $\epsilon = q^2-2$, and then
$$q-1\geq \frac{(q^2-1-q^2+2)(q^2+1)}{(q-1)(q+1)^2} \Leftrightarrow (q-1)^2(q+1)^2\geq q^2+1.$$
This inequality always holds, and moreover, the equality is not attainable for an integer $q\geq2$, hence the inequality is strict.
\end{proof}

\begin{corollary}\label{cor:noMSRD}
    A sum-rank-metric space $\spac$ with $\bfm=\bfn=(2,\dots,2)$, a prime power $q\geq2$, and $t\geq q-1$ does not admit an MSRD code with minimum distance $d=3$.
\end{corollary}

We conclude the section with some computational observations on general sum-rank-metric spaces, in particular those outside of the family specified in Corollary~\ref{cor:noMSRD}. Table~\ref{tab:MSRDDelsarte} lists some examples of sum-rank-metric spaces (or equivalently, graphs) for which Delsarte's LP bound outperforms the Singleton bound of Theorem~\ref{thm:non-induced}, but no other bounds of Theorems~\ref{thm:induced} and \ref{thm:non-induced} nor the Ratio-type bound do. That way, Delsarte's LP bound provides new insight into the non-existence of MSRD codes. We also note that for some parameters listed the same conclusion on non-existence can be reached by computing the Schrijver-type SDP bound discussed in Section~\ref{sec:SchrijverSDPbound}, but Delsarte's LP is often a more computationally efficient choice.

\subsection{Perfect codes}
In this subsection, we investigate the existence of perfect codes in the sum-rank metric. 

Using the estimates on the volume of balls derived in Section \ref{sec:spheresandperfectcodes}, we establish explicit non-existence results for perfect codes in certain families of parameters. 

\subsubsection{$\bfm=(m,1,\dots,1)$ and  $\bfn=(n,1,\dots,1)$}

The next result provides a condition under which the Sphere-Packing bound cannot be met, thereby ruling out the existence of perfect codes with the given parameters in $\spac$ with $\bfm=(m,1,\dots,1)$, $\bfn=(n,1,\dots,1)$.

\begin{theorem}
    Consider the sum-rank-metric space $\spac$ with $\bfm=(m,1,\dots,1)$, $\bfn=(n,1,\dots,1)$, where $m\geq n\geq1$ and $q\geq 2$ is a prime power.
    If
    \[r^2+(m-n-2)r \geq  (t-1)(m-1)+1+ \log_q(r+1)+\log_q\left( {t-1 \choose \lfloor (t-1)/2\rfloor}\right),\]
    then there does not exist any perfect code with minimum distance $d$ such that $r=\lfloor \frac{d-1}2\rfloor$.
\end{theorem}
\begin{proof}
    Suppose that there exists a perfect code $\mathcal{C}$ in $\spac$ with minimum distance $d$ and let $r=\lfloor \frac{d-1}2\rfloor$.
    The Induced Singleton bound (see Theorem \ref{thm:induced}) implies that 
    \begin{equation}\label{eq:condgenISB}
    |\mathcal{C}|\leq q^{m(n+t-1-d+1)}\leq q^{m(n+t-2r-1)}. 
    \end{equation}
    Since $\mathcal{C}$ is perfect then we have
    \[ |\mathcal{C}|V_r(\Fq^{\bfn\times\bfm})=q^{mn+t-1}.  \]
    Using \eqref{eq:condgenISB} and Proposition \ref{prop:boundvolume}, we get that 
    \[q^{mn+t-1} \leq q^{m(n+t-2r-1)} (r+1) {t-1 \choose \lfloor (t-1)/2\rfloor} (q-1)^r q^{(m+n+1)r-r^2+1},  \]
    from which, since the theorem's hypothesis implies $r\geq 1$ and hence $(q-1)^r<q^r$, we get
    \[q^{r^2+(m-n-2)r-(t-1)(m-1)-1}< (r+1) {t-1 \choose \lfloor (t-1)/2\rfloor}, \]
    which implies that 
    \[r^2+(m-n-2)r < (t-1)(m-1)+1+ \log_q(r+1)+\log_q\left( {t-1 \choose \lfloor (t-1)/2\rfloor}\right),\]
    a contradiction to the assertion.
\end{proof}

We can refine the above argument by using the Singleton bound from Theorem \ref{thm:non-induced} when considering the minimum distance greater than $n$.

\begin{theorem}
    Consider the sum-rank-metric space $\spac$ with $\bfm=(m,1,\dots,1)$, $\bfn=(n,1,\dots,1)$, where $m\geq n\geq1$, $q\geq 2$ is a prime power and $d> n$.
    If 
    \[ r^2-(m+n)r > n+1-mn + \log_q(r+1)+\log_q\left( {t-1 \choose \lfloor (t-1)/2\rfloor}\right), \]
    where $r=\lfloor \frac{d-1}2\rfloor$, there do not exist perfect sum-rank-metric codes with minimum distance $d$.
\end{theorem}
\begin{proof}
    We can argue as in the previous theorem. 
    Indeed, let $\C$ be a perfect sum-rank-metric code in $\spac$ with minimum distance $d$.
    Since $d>n$, the Singleton bound (Theorem \ref{thm:non-induced}) reads as follows
    \[ |C|\leq q^{t-d+n}\leq q^{t-2r+n-1}. \]
    Since $\C$ is perfect, then 
    \[ |\mathcal{C}|V_r(\Fq^{\bfn\times\bfm})=q^{mn+t-1}.  \]
    Using the above mentioned bound and Proposition \ref{prop:boundvolume}, we get that 
    \[ q^{mn+t-1}\leq q^{t-2r+n-1}(r+1) {t-1 \choose \lfloor (t-1)/2\rfloor} (q-1)^r q^{(m+n+1)r-r^2+1}, \]
    which implies that
    \[
    q^{r^2-(m+n-1)r+mn-n-1}\leq (r+1) {t-1 \choose \lfloor (t-1)/2\rfloor}(q-1)^r.
    \]
    Since $(q-1)^r\leq q^r$, we obtain
    \[
    q^{r^2-(m+n)r+mn-n-1}\leq (r+1) {t-1 \choose \lfloor (t-1)/2\rfloor},
    \]
    and therefore
    \begin{align*} 
    &r^2-(m+n)r \leq n+1-mn + \log_q(r+1)+\log_q\left( {t-1 \choose \lfloor (t-1)/2\rfloor}\right). 
  \qedhere  \end{align*}
\end{proof}


The above results state that if $d$ is large enough with respect to the length of the blocks, then perfect codes with minimum distance $d$ do not exist.
In the following we will derive some more non existence results based on some congruences.

\begin{proposition}\label{prop:congruence}
    Consider the sum-rank-metric space $\spac$ with $\bfm=(m,1,\dots,1)$, $\bfn=(n,1,\dots,1)$, where $m\geq n\geq1$, $q\geq 2$ is a prime power of the prime $p$. If there exists a non trivial additive perfect sum-rank-metric code with minimum distance $d$, then
    \[ V_{\left\lfloor \frac{d-1}2 \right\rfloor} \equiv 0 \pmod{p}. \]
\end{proposition}
\begin{proof}
    It is is to see that, if $\C$ is an additive perfect code, then the size of $\C$ is a power of $p$ and in order to be perfect we need
    \[ |\mathcal{C}|V_{\left\lfloor \frac{d-1}2 \right\rfloor}(\Fq^{\bfn\times\bfm})=q^{mn+t-1},  \]
    i.e. $V_{\left\lfloor \frac{d-1}2 \right\rfloor}(\Fq^{\bfn\times\bfm})=\frac{q^{mn+t-1}}{|\mathcal{C}|}$ and so the assertion.
\end{proof}

Plugging together Corollary \ref{cor:V1} and Proposition \ref{prop:congruence} we obtain the following result.

\begin{theorem}
    A sum-rank-metric space $\spac$ with $\bfm=(m,1,\dots,1)$, $\bfn=(n,1,\dots,1)$, where $m\geq n\geq1$, while $q\geq 2$ is a prime power, and $t\geq 2$ with $t\not\equiv 1\mod p$, does not admit an additive perfect code with minimum distance $d\in\{3,4\}$.
\end{theorem}
\begin{proof}
    By Corollary \ref{cor:V1}, we have that
    \[ V_1(\Fq^{\bfn\times\bfm})= \frac{q^n-1}{q-1} (q^m-1)+ (t-1)(q-1)+1\equiv -1-t +1+1\pmod{p},\]
    and by Proposition \ref{prop:congruence}.
    \[ V_1(\Fq^{\bfn\times\bfm})\equiv 0 \pmod{p},\] 
    a contradiction to $t\not\equiv 1\mod p$.
\end{proof}

The above result can be extended to the non-additive cases as follows.

In \cite{AKR2024} it has been shown that if $n=1$ or $2$, then in this case the Ratio-type bound does not perform worse that the Sphere-Packing bound if $d=3$. Here, we generalize this result to any integer $n\geq1$, and also note that the Sphere-Packing bound is strictly outperformed by the Ratio-type one.

\begin{theorem}\label{thm:RTSP.oneblock.d=3}
    Let $\Gamma(\spac)$ be a sum-rank-metric graph such that $\bfm=(m,1,\dots,1)$, $\bfn=(n,1,\dots,1)$, with $m\geq n\geq1$, while $t\geq 2$, and $q\geq 2$ is a prime power. Then the Ratio-type bound does not perform worse than the Sphere-Packing bound for minimum distance $d=3$. Moreover, if additionally $t\not\equiv 1\mod q$, then the Ratio-type bound strictly outperforms the Sphere-Packing bound.
\end{theorem}

\begin{proof}
    We first note that the case $n=1$ is shown by \cite[Lemma 25]{AKR2024}. For the rest of the proof, it is assumed that $n\geq 2$.
    
    From \cite[Theorem 24]{AKR2024} we know that in our case, the Ratio-type bound reduces to
    \begin{align*}
        A(q,\bfn,\bfm,3)&\leq\frac{q^{mn+t-1}(q-1)}{(q^m-1)(q^n-1)+(t-1)(q-1)^2+(\varepsilon+1)(q-1)}\cdot \\
        &\cdot\frac{(q^m-1)(q^n-1)+(t-1)(q-1)^2+(\varepsilon+1)(\varepsilon-q+1)(q-1)}{(q^m-1)(q^n-1)+(t-2)(q-1)^2+\varepsilon(q-1)},
    \end{align*}
    where $\varepsilon=(t-1)\mod q$.

    On the other hand, the Sphere-Packing bound (Theorem \ref{thm:non-induced} with Corollary \ref{cor:V1}) implies that
    $$A(q,\bfn,\bfm,3)\leq\frac{q^{mn+t-1}(q-1)}{(q^m-1)(q^n-1)+(t-1)(q-1)^2+q-1}.$$

    Thus, to prove the claim one needs to show that
    \begin{multline*}
        \frac{q^{mn+t-1}(q-1)}{(q^m-1)(q^n-1)+(t-1)(q-1)^2+(\varepsilon+1)(q-1)} \cdot \\
        \cdot\frac{(q^m-1)(q^n-1)+(t-1)(q-1)^2+(\varepsilon+1)(\varepsilon-q+1)(q-1)}{(q^m-1)(q^n-1)+(t-2)(q-1)^2+\varepsilon(q-1)} \leq\\
        \leq\frac{q^{mn+t-1}(q-1)}{(q^m-1)(q^n-1)+(t-1)(q-1)^2+q-1}.
    \end{multline*}

    By dividing both sides by a positive value $q^{mn+t-1}(q-1)$, we get an equivalent inequality:
    \begin{multline*}
        \frac{1}{(q^m-1)(q^n-1)+(t-1)(q-1)^2+(\varepsilon+1)(q-1)} \cdot \\
        \cdot\frac{(q^m-1)(q^n-1)+(t-1)(q-1)^2+(\varepsilon+1)(\varepsilon-q+1)(q-1)}{(q^m-1)(q^n-1)+(t-2)(q-1)^2+\varepsilon(q-1)} \leq\\
        \leq\frac{1}{(q^m-1)(q^n-1)+(t-1)(q-1)^2+q-1}.
    \end{multline*}
    Moreover, since the denominators on both sides are clearly positive values, we can multiply by them and further simplify to
    $$\varepsilon (\varepsilon - q) (q-1) \left((q^m-1)(q^n-1) + (t-1)(q-1)^2\right) \leq 0.$$
    If $\varepsilon=0$, then the equality is achieved, and the Ratio-type and Sphere-Packing bounds coincide, satisfying the theorem's claim. Alternatively, if $\varepsilon>0$, then we have $\varepsilon (\varepsilon - q) (q-1)<0$, we derive
    $$t > -\frac{(q^m-1)(q^n-1)}{(q-1)^2}+1,$$
    which is clearly true under the claim's assumption that $t$ and $q$ are at least $2$ while $m\geq n\geq 1$.
\end{proof}

Theorem~\ref{thm:RTSP.oneblock.d=3} gives an immediate non-existence result for perfect codes.

\begin{corollary}
    A sum-rank-metric space $\spac$ with $\bfm=(m,1,\dots,1)$, $\bfn=(n,1,\dots,1)$, where $m\geq n\geq1$, while $q\geq 2$ is a prime power, and $t\geq 2$ with $t\not\equiv 1\mod q$, does not admit a perfect code with minimum distance $d=3$.
\end{corollary}

\subsubsection{$\bfm=\bfn=(2,2,\dots,2)$}

Based on the result of Theorem~\ref{thm:RT22} we can compare the Ratio-type bound with the Sphere-Packing bound in case $d=3$.

\begin{theorem}
    Let $\spac$ be the sum-rank-metric space with $\bfm=\bfn=(2,\dots,2)$ for $t\geq2$ and a prime power $q\geq2$. Then the Ratio-type bound does not perform worse than the Sphere-Packing bound for minimum distance $d=3$. Moreover, if additionally $t(q+1)\not\equiv 1\mod q^2$, then the Ratio-type bound strictly outperforms the Sphere-Packing bound.
\end{theorem}

\begin{proof}
    Under the conditions specified, the Sphere-Packing bound takes form
    $$A(q,\bfn,\bfm,3)\leq \frac{q^{4t}}{t (q^2 - 1) (q + 1)+1},$$
    so to prove the claim, we need to prove the inequality
    {\small{
    $$\frac{q^{4t}}{t (q^2 - 1) (q + 1)+1} \geq q^{4t}\frac{t(q^2-1)(q+1)-(1+\epsilon)(q^2-1-\epsilon)}{(t(q^2-1)(q+1)+1+\epsilon)(t(q^2-1)(q+1)+1+\epsilon-q^2)},$$
    }}
    where the right-hand side corresponds to the value of the Ratio-type bound derived in Theorem~\ref{thm:RT22}. After standard algebraic manipulations, the inequality can be rewritten as
    $$\frac{\epsilon (q+1)^2 (q-1) (q^2-\epsilon) t}{((q+1)^2 (q-1) t+1) ((q+1)^2 (q-1) t + \epsilon +1) ((q^2-1) (qt+t-1)+\epsilon)} \geq 0.$$
    Clearly, every factor of the left-hand side expression is positive under the assumptions $q\geq2$ and $t\geq2$, except for $\epsilon$, which is non-negative in general and is equal to zero if and only if $t(q+1)\mod q^2$.
\end{proof}

\begin{corollary}
    A sum-rank-metric space $\spac$ with $\bfm=\bfn=(2,\dots,2)$, a prime power $q\geq2$, and $t\geq2$ with $t(q+1)\not\equiv 1\mod q^2$, does not admit a perfect code with minimum distance $d=3$.
\end{corollary}

\begin{remark}
The above non-existence result complements recent results on perfect codes in the sum-rank metric, such as those reported in~\cite[Proposition 5.1 and Theorems 6.2 and 6.3]{PRZ2025}. In particular, the approach taken here, based on comparing the Ratio-type bound with the Sphere-Packing bound, provides an alternative, spectral method for ruling out perfect codes in specific parameter regimes not previously covered in \cite{PRZ2025}. 
\end{remark}

\section{Concluding remarks}\label{sec:concludingremarks}
We investigated upper bounds on the cardinality of sum-rank-metric codes through a combination of spectral, linear programming, and semidefinite programming techniques. By exploiting the connection between sum-rank-metric spaces and graph theory, we developed and compared several bounds, highlighting their relative strengths in different parameter regimes.
We also showed a new application of eigenvalue based bounds arising from sum-rank-metric graphs, illustrating their power in providing several non-existence results for perfect sum-rank-metric codes.

In particular, we first showed a detailed comparison between the Delsarte linear programming bound and the Ratio-type bound, establishing their equivalence in the rank-metric case and clarifying their behavior in more general settings. Furthermore, we introduced a new semidefinite programming (SDP) bound for the sum-rank metric (the so-called Schrijver-type SDP bound), and showed, via computational results, that it improves upon previously known bounds in several instances. Finally, and as a consequence of our bounds (both the eigenvalue bounds as the new Schrijver-type SDP bound), we obtained multiple non-existence results for both maximum sum-rank distance (MSRD) codes and perfect codes. These results illustrate how upper bounds based on algebraic graph theory and optimization techniques can be used not only to efficiently estimating the size of optimal codes, but also to exclude their existence in certain parameter regimes.

Several directions for future research remain open. On the theoretical side, it would be interesting to further investigate the gap between different bounds in the general sum-rank setting and to better understand the regimes in which SDP techniques provide a significant advantage. On the computational side, improving the scalability of SDP methods, for example through analytical symmetry reduction, could allow the study of larger parameter sets. 

\subsection*{Acknowledgements} 

Aida Abiad is supported by NWO (Dutch Research Council) through the grant VI.Vidi.213.085. 
Ferdinando Zullo was partially supported by the Italian National Group for Algebraic and Geometric Structures and their Applications (GNSAGA - INdAM).
This work has been partially written during the Opera 2026 conference in Bordeaux and so we acknowledge the support from the International Research Laboratory LYSM in partnership between CNRS and INdAM. Ferdinando Zullo is very grateful for the hospitality of Department of Mathematics and Computer Science of
Eindhoven University of Technology, where he spent a week as a visiting researcher.

\bibliographystyle{abbrv}
\bibliography{references}

\newpage
\appendix
\section{Tables}
This appendix shows all the tables referenced throughout the paper, providing a view of the data and the obtained computational results for ease of consultation.

{\footnotesize{
\begin{table}[h]
    \centering
\makebox[\textwidth][c]{%
$\begin{array}{|ccc|cc||ccc|cc||ccc|cc|} \hline
t & q & d & \text{DLP=RT} & \text{S} &
t & q & d & \text{DLP=RT} & \text{S} &
t & q & d & \text{DLP=RT} & \text{S} \\ \hline
5 & 3 & 3 & 18 & 27 &
9 & 2 & 4 & 25 & 64 &
12 & 3 & 5 & 1562 & 6561 \\
6 & 2 & 3 & 8 & 16 &
9 & 3 & 3 & 937 & 2187 &
12 & 3 & 6 & 729 & 2187 \\
6 & 3 & 3 & 48 & 81 &
9 & 3 & 4 & 340 & 729 &
13 & 2 & 4 & 292 & 1024 \\
6 & 3 & 4 & 18 & 27 &
9 & 4 & 3 & 9362 & 16384 &
13 & 2 & 11 & 2 & 8 \\
6 & 4 & 3 & 179 & 256 &
9 & 4 & 4 & 2340 & 4096 &
14 & 2 & 3 & 1024 & 4096 \\
7 & 2 & 3 & 16 & 32 &
9 & 4 & 5 & 614 & 1024 &
14 & 2 & 12 & 2 & 8 \\
7 & 2 & 4 & 8 & 16 &
10 & 2 & 3 & 85 & 256 &
15 & 2 & 3 & 2048 & 8192 \\
7 & 2 & 5 & 3 & 8 &
10 & 3 & 3 & 2811 & 6561 &
15 & 2 & 4 & 1024 & 4096 \\
7 & 3 & 3 & 145 & 243 &
10 & 3 & 4 & 937 & 2187 &
15 & 2 & 13 & 2 & 8 \\
7 & 3 & 4 & 48 & 81 &
10 & 3 & 5 & 243 & 729 &
16 & 2 & 3 & 3640 & 16384 \\
7 & 4 & 3 & 614 & 1024 &
11 & 2 & 3 & 170 & 512 &
16 & 2 & 4 & 2048 & 8192 \\
7 & 4 & 4 & 179 & 256 &
11 & 2 & 4 & 85 & 256 &
16 & 2 & 5 & 425 & 4096 \\
7 & 5 & 3 & 2291 & 3125 &
11 & 2 & 9 & 2 & 8 &
16 & 2 & 14 & 2 & 8 \\
8 & 2 & 3 & 25 & 64 &
11 & 3 & 3 & 7029 & 19683 &
17 & 2 & 4 & 3640 & 16384 \\
8 & 2 & 4 & 16 & 32 &
11 & 3 & 4 & 2811 & 6561 &
17 & 2 & 6 & 425 & 4096 \\
8 & 2 & 6 & 3 & 8 &
11 & 3 & 5 & 729 & 2187 &
17 & 2 & 15 & 2 & 8 \\
8 & 3 & 3 & 340 & 729 &
11 & 3 & 6 & 243 & 729 &
18 & 2 & 3 & 13107 & 65536 \\
8 & 4 & 3 & 2340 & 4096 &
12 & 2 & 3 & 292 & 1024 &
18 & 2 & 16 & 2 & 8 \\
8 & 4 & 4 & 614 & 1024 &
12 & 2 & 4 & 170 & 512 &
19 & 2 & 3 & 26214 & 131072 \\
8 & 5 & 3 & 9672 & 15625 &
12 & 2 & 10 & 2 & 8 &
19 & 2 & 4 & 13107 & 65536 \\
8 & 5 & 4 & 2291 & 3125 &
12 & 3 & 4 & 7029 & 19683 &
19 & 2 & 17 & 2 & 8 \\ \hline
\end{array}$
}
    \caption{List of parameters of Hamming graphs on at most 1000000 vertices for which the optima of Delsarte's and the Ratio-type coincide with each other (denoted by DLP and RT, respectively), but not with the Singleton bound (denoted by S). Examples covered by Theorem~\ref{thm:Ham_q=2_d=t-1} are excluded.}
    \label{tab:Ham_Del_RT}
\end{table}
}}

\begin{table*}[!htbp]
{
    \centering
    {\scriptsize
\makebox[\textwidth][c]{%
$\begin{array}{|ccllc|c|ccccc|}
\hline
t & q & {\bf n} & {\bf m} & d & |V| & \alpha_{d-1} & \vartheta_{d-1} & \text{RT}_{d-1} & \text{DLP}_{d} & \text{SDP}_{d} \\
\hline
 2 & 2 &             (2, 2)            &           (2, 2)     & 3       & 256  &  9 & 10 & 11 & 10 & \bf{9} \\  
 3 & 2 &           (2, 2, 1)           &         (2, 2, 1)    & 3       & 512   & 18 & 20 & 25 &  20 & \bf{19} \\ 
 3 & 2 &           (2, 2, 1)           &           (2, 2, 1)    & 4       & 512   & 5     & 8 & 10 & 6 & \bf{5} \\
 4 & 2 &          (2, 1, 1, 1)      &  (2, 2, 2, 1) & 3        & 512 & \text{time} & 24  & 28 & 24 & \bf{23} \\  
 4 & 2 &      (2, 1, 1, 1)      &  (2, 2, 2, 2) & 4        & 1024 & \text{time} &\text{time} & 18 & 10 & \bf{9} \\
 4 & 2 & (2, 2, 1, 1) & (2, 2, 1, 1) & 3 & 1024 & \text{time} &\text{time} & 46 & 40 & \bf{39} \\
 4 & 2 & (2, 2, 1, 1) & (2, 2, 1, 1) & 4 & 1024 & \text{time} &\text{time} & 19 & 12 & \bf{11} \\
 4 & 3 & (1, 1, 1, 1) & (2, 2, 1, 1) & 3 & 729 & \text{time} & 9 & 34 & 9 & \bf{3} \\
 5 & 2 & (2, 1, 1, 1, 1) & (2, 2, 2, 1, 1) & 3 & 1024 & \text{time} &\text{time} & 56 & 49 & \bf{48} \\
 5 & 2 & (2, 1, 1, 1, 1) & (2, 2, 2, 1, 1) & 4 & 1024 & \text{time} &\text{time} & 22 & 13 & \bf{12} \\
 6 & 2 &     (1, 1, 1, 1, 1, 1) &  (2, 1, 1, 1, 1, 1)   & 3        & 128  &   9  & 12 & 12   & 12 & \bf{11} \\
 6 & 2 & (2, 1, 1, 1, 1, 1) & (2, 1, 1, 1, 1, 1) & 4 & 512 & 9 & 12 & 16 & 12 & \bf{11} \\
 6 & 2 & (2, 1, 1, 1, 1, 1) & (2, 2, 1, 1, 1, 1) & 3 & 1024 & \text{time} &\text{time} & 56 & 56 & \bf{53} \\
 6 & 2 & (2, 1, 1, 1, 1, 1) & (2, 2, 1, 1, 1, 1) & 5 & 1024 & \text{time} &\text{time} & 11 & 6 & \bf{5} \\
 7 & 2 &  (1, 1, 1, 1, 1, 1, 1)  &  (2, 1, 1, 1, 1, 1, 1) & 3 & 256  &  \text{time}  & 25 &  25  & 21 & \bf{20} \\
 7 & 2 & (1, 1, 1, 1, 1, 1, 1) & (2, 2, 1, 1, 1, 1, 1) & 4 & 512 & 9 & 12 & 19 & 12 & \bf{11} \\
 7 & 2 & (1, 1, 1, 1, 1, 1, 1) & (3, 2, 1, 1, 1, 1, 1) & 4 & 1024 & \text{time} &\text{time} & 30 & 12 & \bf{11} \\
 7 & 2 & (2, 1, 1, 1, 1, 1, 1) & (2, 1, 1, 1, 1, 1, 1) & 4 & 1024 & \text{time} &\text{time} & 30 & 21 & \bf{20} \\
 8 & 2 &    (1, 1, 1, 1, 1, 1, 1, 1)   &    (2, 1, 1, 1, 1, 1, 1, 1) & 3   & 512  &  \text{time}  & 42 &  42  & 42 & \bf{41} \\
\hline
\end{array}$
}
}
}
\caption{Some sum-rank-metric graphs on $|V|\leq2000$ vertices with the exact $(d-1)$-independence number $\alpha_{d-1}$, the Lov\'asz theta-$(d-1)$ number $\vartheta_{d-1}$, Ratio-type bound $\mathrm{RT}_{d-1}$, Delsarte's LP bound $\mathrm{DLP}_{d}$, and the Schrijver-type SDP bound $\mathrm{SSDP}_{d}$, and the Schrijver-type SDP bound strictly improving the result of Delsarte's LP bound. We write `time' whenever the computation time exceeds $2$ hours on a standard laptop.}
    \label{tab:sdp}
\end{table*}

\begin{table}[]
{\scriptsize
    \centering
\makebox[\textwidth][c]{%
    $\begin{array}{|cc|cc|c|c|cc|cccc|cccc|} \hline
    t & q & {\bf n} & {\bf m} & d & |V| & \text{RT} & \text{DLP} & \text{iS} & \text{iH} & \text{iP} & \text{iE} & \text{S} & \text{SP} & \text{PSP} & \text{TD} \\ \hline
2 & 2 & (3, 2) & (3, 2) & 3 & 8192 & 131 & 109 & 512 & 910 & 0 & 2222 & 128 & 138 & 138 & 0 \\
2 & 2 & (3, 2) & (3, 2) & 4 & 8192 & 31 & 14 & 64 & 910 & 0 & 272 & 16 & 138 & 40 & 0 \\
2 & 2 & (3, 2) & (3, 3] & 4 & 32768 & 93 & 36 & 64 & 910 & 0 & 272 & 64 & 461 & 141 & 0 \\
3 & 2 & (2, 2, 1) & (2, 2, 1) & 4 & 512 & 10 & 6 & 16 & 64 & 16 & 27 & 8 & 25 & 18 & 0 \\
3 & 2 & (3, 2, 1) & (3, 2, 1) & 3 & 16384 & 273 & 219 & 4096 & 6096 & 0 & 15362 & 256 & 273 & 273 & 0 \\
3 & 2 & (3, 2, 1) & (3, 2, 1) & 4 & 16384 & 72 & 28 & 512 & 6096 & 0 & 1768 & 32 & 273 & 81 & 0 \\
3 & 2 & (3, 2, 1) & (3, 2, 2) & 4 & 32768 & 134 & 57 & 512 & 6096 & 0 & 1768 & 64 & 528 & 163 & 0 \\
3 & 2 & (3, 2, 1) & (3, 2, 2) & 5 & 32768 & 42 & 11 & 64 & 336 & 0 & 240 & 16 & 33 & 78 & 0 \\
3 & 2 & (2, 2, 2) & (3, 2, 2) & 4 & 16384 & 115 & 59 & 512 & 6096 & 0 & 1768 & 64 & 409 & 146 & 0 \\
4 & 2 & (2, 1, 1, 1) & (2, 2, 2, 1) & 4 & 512 & 11 & 6 & 16 & 64 & 16 & 27 & 8 & 30 & 32 & 12 \\
4 & 2 & (3, 1, 1, 1) & (3, 2, 2, 2) & 3 & 32768 & 526 & 493 & 4096 & 6096 & 0 & 15362 & 512 & 555 & 555 & 0 \\
4 & 2 & (3, 1, 1, 1) & (3, 2, 2, 2) & 5 & 32768 & 50 & 14 & 64 & 336 & 0 & 240 & 16 & 39 & 146 & 0 \\
4 & 2 & (2, 2, 1, 1) & (2, 2, 1, 1) & 4 & 1024 & 19 & 12 & 64 & 215 & 0 & 119 & 16 & 48 & 36 & 0 \\
4 & 2 & (2, 2, 1, 1) & (2, 2, 2, 1) & 4 & 2048 & 32 & 16 & 64 & 215 & 0 & 119 & 32 & 89 & 73 & 0 \\
4 & 2 & (2, 2, 1, 1) & (2, 2, 2, 1) & 5 & 2048 & 12 & 5 & 16 & 26 & 10 & 19 & 8 & 10 & 9 & 12 \\
4 & 2 & (2, 2, 1, 1) & (3, 2, 2, 2) & 5 & 16384 & 43 & 11 & 64 & 336 & 0 & 240 & 16 & 35 & 16 & 0 \\
4 & 2 & (3, 2, 1, 1) & (3, 2, 1, 1) & 3 & 32768 & 528 & 438 & 32768 & 32768 & 0 & 32768 & 512 & 537 & 537 & 0 \\
4 & 2 & (3, 2, 1, 1) & (3, 2, 1, 1) & 4 & 32768 & 151 & 57 & 4096 & 32768 & 0 & 11904 & 64 & 537 & 163 & 0 \\
4 & 2 & (2, 2, 2, 1) & (2, 2, 2, 1) & 6 & 8192 & 12 & 5 & 16 & 77 & 8 & 14 & 8 & 25 & 18 & 12 \\
4 & 2 & (2, 2, 2, 1) & (3, 2, 2, 1) & 4 & 32768 & 229 & 115 & 4096 & 32768 & 0 & 11904 & 128 & 799 & 292 & 0 \\
4 & 2 & (2, 2, 2, 1) & (3, 2, 2, 1) & 6 & 32768 & 29 & 6 & 64 & 1943 & 0 & 211 & 8 & 55 & 18 & 24 \\
5 & 2 & (2, 1, 1, 1, 1) & (2, 2, 2, 1, 1) & 4 & 1024 & 22 & 13 & 64 & 215 & 0 & 119 & 16 & 56 & 64 & 0 \\
5 & 2 & (2, 1, 1, 1, 1) & (2, 2, 2, 2, 1) & 4 & 2048 & 36 & 20 & 64 & 215 & 0 & 119 & 32 & 102 & 128 & 0 \\
5 & 2 & (2, 1, 1, 1, 1) & (2, 2, 2, 2, 1) & 5 & 2048 & 14 & 6 & 16 & 26 & 10 & 19 & 8 & 13 & 11 & 8 \\
5 & 2 & (2, 2, 1, 1, 1) & (2, 2, 1, 1, 1) & 4 & 2048 & 38 & 24 & 256 & 744 & 0 & 407 & 32 & 93 & 73 & 0 \\
5 & 2 & (2, 2, 1, 1, 1) & (2, 2, 1, 1, 1) & 5 & 2048 & 13 & 7 & 64 & 77 & 0 & 99 & 8 & 11 & 9 & 0 \\
5 & 2 & (2, 2, 1, 1, 1) & (2, 2, 2, 1, 1) & 4 & 4096 & 64 & 32 & 256 & 744 & 0 & 407 & 64 & 170 & 146 & 0 \\
5 & 2 & (2, 2, 1, 1, 1) & (2, 2, 2, 1, 1) & 5 & 4096 & 22 & 11 & 64 & 77 & 0 & 99 & 16 & 19 & 17 & 0 \\
5 & 2 & (2, 2, 1, 1, 1) & (2, 2, 2, 2, 1) & 6 & 8192 & 15 & 5 & 16 & 77 & 8 & 14 & 8 & 31 & 32 & 8 \\
5 & 2 & (2, 2, 1, 1, 1) & (3, 2, 2, 2, 1) & 6 & 32768 & 38 & 6 & 64 & 1943 & 0 & 211 & 8 & 65 & 32 & 10 \\
5 & 2 & (2, 2, 2, 1, 1) & (2, 2, 2, 1, 1) & 6 & 16384 & 23 & 9 & 64 & 236 & 0 & 79 & 16 & 47 & 36 & 0 \\
5 & 2 & (2, 2, 2, 1, 1) & (2, 2, 2, 2, 1) & 6 & 32768 & 36 & 12 & 64 & 236 & 0 & 79 & 32 & 81 & 73 & 0 \\
6 & 2 & (2, 1, 1, 1, 1, 1) & (2, 1, 1, 1, 1, 1) & 4 & 512 & 16 & 12 & 256 & 512 & 0 & 407 & 16 & 34 & 32 & 0 \\
6 & 2 & (2, 1, 1, 1, 1, 1) & (2, 2, 1, 1, 1, 1) & 5 & 1024 & 11 & 6 & 64 & 77 & 0 & 99 & 8 & 9 & 8 & 8 \\
6 & 2 & (2, 1, 1, 1, 1, 1) & (2, 2, 2, 1, 1, 1) & 4 & 2048 & 42 & 26 & 256 & 744 & 0 & 407 & 32 & 107 & 128 & 0 \\
6 & 2 & (2, 1, 1, 1, 1, 1) & (2, 2, 2, 2, 1, 1) & 4 & 4096 & 70 & 40 & 256 & 744 & 0 & 407 & 64 & 195 & 256 & 0 \\
6 & 2 & (2, 1, 1, 1, 1, 1) & (2, 2, 2, 2, 1, 1) & 5 & 4096 & 26 & 12 & 64 & 77 & 0 & 99 & 16 & 23 & 21 & 0 \\
6 & 2 & (2, 1, 1, 1, 1, 1) & (2, 2, 2, 2, 2, 1) & 5 & 8192 & 41 & 16 & 64 & 77 & 0 & 99 & 32 & 38 & 36 & 0 \\
6 & 2 & (2, 2, 1, 1, 1, 1) & (2, 2, 1, 1, 1, 1) & 4 & 4096 & 72 & 48 & 1024 & 2621 & 0 & 1419 & 64 & 178 & 146 & 0 \\
6 & 2 & (2, 2, 1, 1, 1, 1) & (2, 2, 1, 1, 1, 1) & 5 & 4096 & 24 & 14 & 256 & 236 & 0 & 366 & 16 & 21 & 18 & 0 \\
6 & 2 & (2, 2, 1, 1, 1, 1) & (2, 2, 1, 1, 1, 1) & 6 & 4096 & 12 & 6 & 64 & 236 & 0 & 79 & 8 & 21 & 16 & 8 \\
6 & 2 & (2, 2, 1, 1, 1, 1) & (2, 2, 2, 1, 1, 1) & 5 & 8192 & 39 & 21 & 256 & 236 & 0 & 366 & 32 & 34 & 32 & 0 \\
6 & 2 & (2, 2, 1, 1, 1, 1) & (2, 2, 2, 1, 1, 1) & 6 & 8192 & 18 & 7 & 64 & 236 & 0 & 79 & 8 & 34 & 32 & 16 \\
6 & 2 & (2, 2, 1, 1, 1, 1) & (3, 2, 2, 1, 1, 1) & 5 & 32768 & 90 & 31 & 4096 & 11740 & 0 & 14226 & 32 & 70 & 32 & 0 \\
6 & 2 & (2, 2, 1, 1, 1, 1) & (2, 2, 2, 2, 1, 1) & 6 & 16384 & 28 & 10 & 64 & 236 & 0 & 79 & 16 & 57 & 64 & 0 \\
6 & 2 & (2, 2, 1, 1, 1, 1) & (2, 2, 2, 2, 2, 1) & 6 & 32768 & 44 & 12 & 64 & 236 & 0 & 79 & 32 & 96 & 128 & 0 \\
6 & 2 & (2, 2, 2, 1, 1, 1) & (2, 2, 2, 1, 1, 1) & 6 & 32768 & 42 & 17 & 256 & 744 & 0 & 331 & 32 & 87 & 73 & 0 \\
6 & 2 & (2, 2, 2, 1, 1, 1) & (2, 2, 2, 1, 1, 1) & 7 & 32768 & 17 & 6 & 64 & 100 & 28 & 52 & 8 & 14 & 9 & 16 \\
7 & 2 & (2, 1, 1, 1, 1, 1, 1) & (2, 2, 2, 1, 1, 1, 1) & 4 & 4096 & 85 & 53 & 1024 & 2621 & 0 & 1419 & 64 & 204 & 256 & 0 \\
7 & 2 & (2, 1, 1, 1, 1, 1, 1) & (2, 2, 2, 2, 1, 1, 1) & 4 & 8192 & 141 & 80 & 1024 & 2621 & 0 & 1419 & 128 & 372 & 512 & 0 \\
7 & 2 & (2, 1, 1, 1, 1, 1, 1) & (2, 2, 2, 2, 1, 1, 1) & 5 & 8192 & 45 & 24 & 256 & 236 & 0 & 366 & 32 & 42 & 39 & 0 \\
7 & 2 & (2, 1, 1, 1, 1, 1, 1) & (2, 2, 2, 2, 2, 1, 1) & 5 & 16384 & 75 & 32 & 256 & 236 & 0 & 366 & 64 & 69 & 68 & 0 \\
7 & 2 & (2, 1, 1, 1, 1, 1, 1) & (2, 2, 2, 2, 2, 1, 1) & 6 & 16384 & 32 & 10 & 64 & 236 & 0 & 79 & 16 & 69 & 21 & 20 \\
7 & 2 & (2, 1, 1, 1, 1, 1, 1) & (2, 2, 2, 2, 2, 2, 1) & 6 & 32768 & 47 & 14 & 64 & 236 & 0 & 79 & 32 & 116 & 36 & 0 \\
\hline
    \end{array}$
}
}
    \caption{\footnotesize List of sum-rank-metric graphs on $|V|\leq50000$ vertices with $t\leq 7$ for which Delsarte's LP bound (DLP) rules out the existence of an MSRD code while the Ratio-type bound (RT) or the coding bounds of Theorems~\ref{thm:induced}, \ref{thm:non-induced-singleton-td} and \ref{thm:non-induced} do not. Here, iS, iH, iP, and iE stand for induced Singleton, Hamming, Plotkin, and Elias bounds, while S, SP, PSP, and TD stand for Singleton, Sphere-Packing, Projective Sphere-Packing, and Total Distance bounds, respectively. We write `$0$' if the bound is not applicable for the choice of parameters.}
    \label{tab:MSRDDelsarte}
\end{table}

\end{document}